\documentclass[%
 reprint,
 amsmath,amssymb,
 aps, physrev,
]{revtex4-2}

\usepackage{graphicx}
\usepackage{dcolumn}
\usepackage{bm}

\usepackage[most]{tcolorbox}
\usepackage{enumitem}
\usepackage{mathtools}

\usepackage{amsthm}
 
\newtheorem{theorem}{Theorem}
\newtheorem{definition}[theorem]{Definition}

\usepackage{graphicx}
\usepackage{subcaption}

\newcommand{\protocol}[1]{%
  \noindent\rule{\linewidth}{0.5pt}
  
  \noindent\textbf{Protocol:}~#1
  
  \noindent\rule{\linewidth}{0.5pt}
}

\newcounter{experiment}

\newcommand{\experiment}[1]{
  \refstepcounter{experiment}
  \noindent\rule{\linewidth}{0.5pt}

  \noindent\textbf{Experiment~\theexperiment}~#1

  \noindent\rule{\linewidth}{0.5pt}
}

\begin{document}

\preprint{APS/123-QED}

\title{\textbf{Traceable Quantum Secret Sharing} 
}%

\author{Mengyan Qiao} 
\affiliation{School of Information Science and Engineering, East China University of Science and Technology, Shanghai 200237, China}

\author{Jian Long}
\email{longjian@ecust.edu.cn}
\affiliation{School of Information Science and Engineering, East China University of Science and Technology, Shanghai 200237, China}

\author{Yu Li}
\email{liyu@ecust.edu.cn}
\affiliation{School of Information Science and Engineering, East China University of Science and Technology, Shanghai 200237, China}

\author{Zhu Cao}
\email{caozhu@tongji.edu.cn}

\affiliation{College of Electronics and Information Engineering, Tongji University, Shanghai 201804, China}
\affiliation{Shanghai Research Institute for Intelligent Autonomous Systems, Tongji University, Shanghai 201210, China}
 
\date{\today}

\begin{abstract}
Quantum secret sharing protects distributed quantum information against unauthorized reconstruction, yet existing schemes provide no intrinsic accountability when shares are leaked or forged. This missing capability is particularly relevant in realistic distributed settings, where compromised participants may either expose valid shares or maliciously interfere with reconstruction. Here we establish a framework for traceable quantum secret sharing under these two threat models. We introduce leakage-verifiable and traceable quantum secret sharing (LVT-QSS), which attributes leaked shares through an explicit leakage verification mechanism, and malicious-traceable and accountable quantum secret sharing (MTA-QSS), which identifies participants who submit forged shares during reconstruction. We formulate security definitions for both settings and prove that the proposed protocols satisfy the corresponding guarantees against the relevant adversaries. Numerical simulations under per-share depolarizing noise further quantify the robustness of the proposed schemes and identify verification as the dominant noise sensitive stage. These results incorporate accountability into quantum secret sharing and provide a foundation for traceable distributed quantum cryptography.

\end{abstract}

\maketitle


\section{Introduction}
\label{sec:intro}

Quantum secret sharing (QSS) is a central primitive for protecting distributed quantum information and has broad relevance to multiparty quantum cryptography~\cite{25_bozzio2025quantum}, quantum networks~\cite{26_zheng2026large}, and distributed quantum computation~\cite{27_crepeau2002secure}. In a threshold secret sharing scheme, a dealer distributes a secret among $n$ participants such that any authorized set of at least $k$ parties can reconstruct it, whereas any unauthorized set of fewer than $k$ parties learns nothing about the secret. Since the first QSS protocol based on Greenberger--Horne--Zeilinger states was introduced by Hillery \emph{et al.}~\cite{3_hillery1999quantum}, and threshold constructions based on quantum error-correcting codes (QECCs) were developed by Cleve \emph{et al.}~\cite{4_cleve1999share}, QSS has become a basic tool for secure distributed quantum information processing.

Existing QSS schemes are primarily designed to guarantee correctness and confidentiality: authorized sets can recover the secret, while unauthorized sets obtain no information about it. In realistic distributed settings, however, this is not always sufficient. Valid shares may be leaked, resold, or otherwise exposed by compromised participants, and malicious parties may interfere with reconstruction by submitting forged shares. In either case, conventional QSS provides no intrinsic mechanism to identify the source of compromise or to assign accountability. As a result, QSS remains a confidentiality-oriented primitive even in settings where attribution of leakage or misconduct is operationally essential.

This missing accountability layer is particularly significant in quantum settings. QSS is intended for scenarios in which quantum information is protected across multiple locations or parties, often in environments where trust is limited and failures may be adversarial rather than merely accidental. In such settings, one would like not only to preserve secrecy below threshold, but also to determine which participants are responsible when valid shares are leaked or when the reconstruction process is deliberately corrupted. This motivates the study of \emph{traceable} QSS, in which attribution and accountability become explicit security objectives.

In the classical setting, this question has recently been formalized through \emph{traceable secret sharing} (TSS). Goyal \emph{et al.}~\cite{5_goyal2021traceable} introduced the basic notion of TSS, aiming to enable provable attribution of leaked shares. Boneh \emph{et al.}~\cite{6_boneh2024traceable} subsequently refined the formal framework and identified stronger security requirements, and later work extended these definitions and constructions to more general access structures~\cite{7_goyal2025traceable}. At a high level, TSS requires that if an unauthorized coalition aggregates its shares into a \emph{reconstruction box} that reconstructs the secret when supplied with sufficiently many additional shares, then an efficient tracing procedure, given only black-box access to this box, can identify the parties responsible for constructing it. Beyond traceability, TSS also requires protections such as non-imputability, which prevents a malicious tracer from falsely implicating honest parties, and tracer secrecy, which requires that tracing itself not compromise secrecy; in particular, access to the tracing key should reveal no additional information about the secret.

Despite this progress, existing work on TSS has focused almost exclusively on classical secret sharing. Extending traceability to QSS is nontrivial. The shared objects may themselves be quantum states, reconstruction is governed by quantum operations, and accountability must be formulated in a way that is compatible with distributed quantum information processing. To the best of our knowledge, no existing QSS scheme in the literature provides a systematic notion of traceability together with publicly verifiable accountability. Consequently, tracing leaked quantum shares and identifying malicious participants during reconstruction remain open problems.

In this work, we establish a framework for traceable quantum secret sharing under two realistic threat models: one in which participants follow the protocol but leak or resell valid shares, and another in which malicious participants deviate during reconstruction by submitting forged shares. We make three main contributions: (i) we introduce leakage-verifiable and traceable quantum secret sharing, which attributes leaked shares through an explicit leakage-verification mechanism; (ii) we propose malicious-traceable and accountable quantum secret sharing, which identifies participants who submit forged shares during reconstruction and supports publicly verifiable accountability; and (iii) we formulate security definitions for both settings, prove that the proposed protocols achieve the corresponding security guarantees against the relevant adversaries, and perform numerical simulations under per-share depolarizing noise to evaluate their robustness. Our results extend QSS beyond confidentiality alone by incorporating accountability as a security objective, and provide a foundation for traceable quantum secret sharing and accountable multiparty quantum cryptographic primitives.

\subsection{Threat-model overview}
\label{sub:threatmodel}

We consider two operational threat models for traceable quantum secret sharing.

\paragraph*{Honest-but-leaky model.}
In the first setting, participants follow the protocol honestly, but an unauthorized coalition may leak or resell its valid quantum shares to a third party. These leaked shares can then be embedded into a \emph{reconstruction box} that outputs the secret once it receives sufficiently many additional shares from outside the coalition. The tracing goal is to identify, using only black-box access to this box, the participants whose shares were used to construct it. We focus on coalitions smaller than the reconstruction threshold, since larger coalitions can reconstruct and leak the secret on their own without leaving a meaningful tracing target.

\paragraph*{Malicious-but-accountable model.}
In the second setting, malicious participants actively deviate during reconstruction by submitting forged shares. Here the threat is not passive leakage but corruption of the reconstruction procedure itself. Because the submitted shares may be inconsistent with the prescribed protocol, we do not use the reconstruction-box abstraction in this case. Instead, tracing is performed directly from the auxiliary information generated during reconstruction. The objective is to identify the participants responsible for the forged shares while providing evidence that can be independently verified by other parties. As in the leakage setting, we restrict attention to adversarial sets smaller than the reconstruction threshold, since larger sets can reconstruct and leak the secret without requiring any additional honest shares.

\subsection{Related work}
\label{sub:relatedwork}

Quantum secret sharing has developed along two main directions: constructions based on quantum error-correcting codes (QECCs) and protocols based on multipartite entanglement. In the QECC-based approach, a secret quantum state is encoded into distributed shares and recovered by authorized sets through decoding. Cleve \emph{et al.}~\cite{4_cleve1999share} gave the first general threshold construction, and Gottesman~\cite{8_gottesman2000theory} later developed a broader theoretical framework for QSS. Subsequent work explored variants emphasizing coding efficiency, verifiability, and communication efficiency~\cite{9_ogawa2005quantum,10_lipinska2020verifiable,11_marin2012equivalence,12_senthoor2021universal,13_senthoor2024communication}. In parallel, entanglement-based schemes---beginning with the GHZ-state protocol of Hillery \emph{et al.}~\cite{3_hillery1999quantum}---use multipartite correlations rather than explicit QECC decoding, and have been extended to continuous-variable, measurement-device-independent, and networked settings~\cite{15_kogias2017unconditional,16_wang2019measurement,17_cai2024measurement,18_liu2025measurement,19_zhang2025efficient}.

These works have substantially advanced the confidentiality, robustness, and efficiency of QSS. However, they do not address a different question that becomes important in realistic distributed settings: how to attribute responsibility when valid shares are leaked or when malicious participants submit forged shares during reconstruction. In most existing schemes, leakage is treated either as a generic failure event or as part of an adversarial model dominated by external eavesdropping, rather than as a source of ex post accountability. As a result, current QSS frameworks do not provide systematic notions of traceability, non-imputability, or publicly verifiable attribution in the presence of leaked or forged shares.

Related ideas have recently been developed in the classical setting under the framework of traceable secret sharing~\cite{5_goyal2021traceable,6_boneh2024traceable,7_goyal2025traceable}. To the best of our knowledge, however, no existing QSS scheme provides a corresponding framework for traceability together with publicly verifiable accountability under the threat models considered here. This gap motivates the present work.

\section{Preliminaries}
\label{sec:preliminaries}

\subsection{Notation}
\label{sub:notation}

We consider a threshold setting with parameter $k$ and $m=2k-1$ parties, indexed by
$[m]=\{0,\ldots,m-1\}$. The secret is a quantum state $|\psi\rangle\in\mathcal{H}_\lambda$, where $\mathcal{H}_\lambda$ is a $d$-dimensional Hilbert space and $\lambda$ is the security parameter. We work over the finite field $\mathbb{F}=\mathbb{Z}_q$, where $q$ is a prime satisfying $q\ge \max\{m,d\}$.

We further assume a post-quantum symmetric-key encryption scheme
$\Pi_{\mathsf{SKE}}=(\mathrm{Enc},\mathrm{Dec})$ with message space $\mathbb{F}$ and ciphertext space $\mathcal{C}$, together with a quantum-secure collision-resistant and one-way hash function
$H:\mathbb{F}\to\mathbb{F}$, as defined in Appendix~\ref{appendix:quantumhash}.

\subsection{Quantum information}
\label{sub:quantuminformation}

In this section, we briefly recall several standard notions from quantum information; for a comprehensive treatment, see Ref.~\cite{20_wang_quantum_2012}.

A quantum channel $\mathcal{F}$ is a completely positive and trace-preserving (CPTP) map from register $x$ to register $y$. That is, for any input density operator $\rho_x$, the channel $\mathcal{F}$ outputs a state $\tau_y$ on register $y$, namely $\mathcal{F}(\rho_x)=\tau_y$. When the input is part of a larger joint state $\rho_{xz}$, the channel $\mathcal{F}$ acts nontrivially only on register $x$, producing an output on register $y$ while leaving the other registers, such as $z$, unchanged.

A projector $\Pi$ is a Hermitian operator satisfying $\Pi^2=\Pi$. A projective measurement is specified by a collection of projectors $\{ \Pi_j\}_j$ such that $\sum_j \Pi_j=\mathrm{Id}$. Given a quantum state $\rho$, the probability of obtaining outcome $j$ is $p_j = \mathrm{Tr}(\Pi_j \rho)$, and the corresponding post-measurement state is $\rho_j = \frac{\Pi_j \rho \Pi_j}{p_j}$. For registers $x$ and $y$, the partial trace $\mathrm{Tr}_y$ is the unique linear map satisfying $\mathrm{Tr}_y(A_x \otimes B_y) = \mathrm{Tr}(B_y)\,A_x$ for all operators $A_x$ and $B_y$. For any joint state $\rho_{xy}$, we write $\rho_x := \mathrm{Tr}_y(\rho_{xy})$.

The trace distance between two quantum states $\rho$ and $\tau$, denoted by $\mathrm{Td}(\rho,\tau)$, is defined as
\begin{equation}
    \mathrm{Td}(\rho, \tau) = \frac{1}{2} \lVert \rho - \tau \rVert_1 = \frac{1}{2} \mathrm{Tr} \left( \sqrt{(\rho - \tau)^\dagger (\rho - \tau)} \right),
\end{equation}
where $\lVert \rho-\tau\rVert_1$ is the trace norm, i.e., the sum of the singular values of $\rho-\tau$. The trace distance quantifies the operational distinguishability of two quantum states.

In particular, the Helstrom theorem states that if $\rho$ and $\tau$ are given with equal prior probabilities, then the optimal success probability for distinguishing them is $\frac{1}{2}\bigl(1+\mathrm{Td}(\rho,\tau)\bigr).$ When $\tau$ is pure, i.e., $\tau=\lvert\Phi\rangle\langle\Phi\rvert$, the Fuchs--van de Graaf inequalities imply
\begin{equation}
    1 - \langle \Phi \vert \rho \vert \Phi \rangle \leq \mathrm{Td}(\rho, \tau) \leq \sqrt{1 - \left\lvert \langle \Phi \vert \rho \vert \Phi \rangle \right\rvert^2}. 
\end{equation}

A nonuniform quantum polynomial-time (QPT) machine $\{ \mathcal{A}_{\lambda} \}_{\lambda \in \mathbb{N}}$ is a family of polynomial-size quantum circuits, where each circuit $\mathcal{A}_\lambda$ is provided with a polynomial-size quantum advice state $\lvert\alpha_\lambda\rangle$. Each $\mathcal{A}_\lambda$ induces a CPTP map. More generally, an interactive quantum machine can be modeled as a sequence of quantum channels acting on designated input, output, and internal registers.

Two probability distributions $X$ and $Y$ are computationally indistinguishable if there exists a negligible function $\mathrm{negl}(\lambda)$ such that, for every QPT algorithm $\mathcal{A}_\lambda$, it holds that
\begin{multline}
  \left| \Pr[1 \leftarrow \mathcal{A}_\lambda(x) : x \leftarrow X] - \Pr[1 \leftarrow \mathcal{A}_\lambda(y) : y \leftarrow Y] \right| \\
  \leq \mathrm{negl}(\lambda).
\end{multline}

If the same condition holds for all, possibly unbounded, quantum algorithms, then $X$ and $Y$ are said to be statistically indistinguishable.

\subsection{Technical assumptions}
\label{tech}

\begin{definition}[Sparse-hit property]
\label{def:sparsehit}
Let $\Pi_{\mathsf{SKE}}=(\mathrm{Enc},\mathrm{Dec})$ be a symmetric-key encryption scheme with ciphertext space $\mathcal{C}$. We say that $\Pi_{\mathsf{SKE}}$ satisfies the \emph{sparse-hit property} with parameter $\alpha_{\max}(\lambda)$ if, for every security parameter $\lambda$, every key $k\in\{0,1\}^{\lambda}$, and every polynomial-size subset $S\subseteq\mathcal{C}$,
\begin{equation}
    \Pr_{x}\!\bigl[\mathrm{Dec}(x,k)\in S\bigr]
    \le |S|\,\alpha_{\max}(\lambda),
\end{equation}
where the probability is taken over a random variable $x$ taking values in $\mathcal{C}$.
\end{definition}

\section{RESULTS}
\label{result}
\subsection{Honest-but-leaky model}
\label{sub:honest}

In our protocol, the secret quantum state $|\psi\rangle$ is encoded into $m$ quantum shares $\{|sh_i\rangle\}_{i\in[0,m-1]}$, each associated with a unique tag $\{tag_i\}_{i\in[0,m-1]}$. These tags, stored in the auxiliary registers, serve as identifiers for the shares and are critical for subsequent leakage tracing and evidence verification. A set of at least $k$ shares is sufficient to reconstruct the original secret $|\psi\rangle$.

The leakage tracing procedure proceeds as follows: Assume a reconstruction box $\mathcal{R}$ that contains $f$ leaked shares and, given an additional $k-f$ shares, reconstructs the secret $|\psi\rangle$ along with the associated auxiliary registers. Let $\mathcal{I}_{\mathrm{idx}}\subseteq [m]$ denote the index set of the leaked shares, and let $\mathcal{I}:=\{tag_i \mid i\in \mathcal{I}_{\mathrm{idx}}\}$ be the corresponding tag set. The verifier invokes $\mathcal{R}$, then uses the tracing key to decrypt the auxiliary registers and recover the corresponding tags, obtaining an evidence set $\mathcal{I}'$.

Finally, the verifier checks whether $\mathcal{I}'$ satisfies the protocol’s verification conditions. If successful, $\mathcal{I}'$ is accepted as valid evidence of the leaked shares, identifying the participants responsible for the leakage based on the uniqueness of the tags.

\subsubsection{Security definitions}
\label{subsub:secdefleakage}

We define an LVT-QSS scheme that allows for the secure sharing of a quantum state $|\psi\rangle$ among $m$ participants. The scheme is required to support secret reconstruction from any authorized subset of at least $k$ shares, as well as leakage tracing and evidence verification through auxiliary tag information.

Formally, we model an LVT-QSS scheme by five QPT algorithms,
$\mathrm{Setup}$, $\mathrm{GenSecretShares}$, $\mathrm{SecretRecover}$,
$\mathrm{SecretTrace}$, and $\mathrm{SecretVerify}$.
Whenever convenient, we use assignment notation for their execution, e.g., $(tk,Tags)\leftarrow \mathrm{Setup}(1^\lambda,m)$.

\begin{itemize}
    \item \textbf{Setup}: Given the security parameter $1^\lambda$ and the number of parties $m$, the trusted authority generates a tracing key $tk$ and a tag set $Tags = \{tag_i\}_{i=0}^{m-1}$.

    \item \textbf{GenSecretShares}: Given a quantum state $|\psi\rangle$, the number of parties $m$, a threshold $k$, a tag set $Tags$, and the tracing key $tk$, the dealer generates $m$ quantum shares $\{|sh_i\rangle\}_{i=0}^{m-1}$ and the corresponding verification keys $\{vk_i\}_{i=0}^{m-1}$.

    \item \textbf{SecretRecover}: Given a set of shares $\{|sh_i\rangle\}_{i\in P}$ with $P\subseteq [m]$, the algorithm attempts to reconstruct the secret state $|\psi\rangle$. If reconstruction fails, it outputs $\bot$.

    \item \textbf{SecretTrace}: Given the tracing key $tk$ and a set of quantum shares $\{|sh_i\rangle\}_{i\in Q\subseteq [m]\setminus \mathcal{I}_{\mathrm{idx}}}$ with $|Q|=k-f$ for some $f<k$, and oracle access to a reconstruction device $\mathcal{R}$ built from $f$ leaked shares, the tracer outputs a candidate evidence set $\mathcal{I}'$.

    \item \textbf{SecretVerify}: Given the verification keys $\{vk_i\}_{i=0}^{m-1}$ and a candidate evidence set $\mathcal{I}'$, the verifier outputs a bit $b\in\{0,1\}$ indicating whether $\mathcal{I}'$ satisfies the verification conditions.
\end{itemize}

Correctness requires that any authorized subset of at least $k$ shares recover the original secret state up to negligible error in trace distance.

\begin{definition}[Correctness (informal)]
\label{def:correct1}
An LVT-QSS scheme is correct if any authorized subset $P\subseteq [m]$ with $|P|\ge k$ can recover the secret state with probability at least $1-\mathrm{negl}(\lambda)$. Upon successful recovery, the reconstructed state is within trace distance at most $\mathrm{negl}(\lambda)$ of the original secret state. A fully formal definition is given in Appendix~\ref{def:correct1-formal}.
\end{definition}

An LVT-QSS scheme is required to satisfy three central security properties: secrecy, traceability, and non-imputability. Secrecy is the standard threshold requirement that any subset of fewer than $k$ parties obtains essentially no information about the secret. We formalize this requirement as follows.

\begin{definition}[Secrecy (informal)]
\label{def:secrecy}
An LVT-QSS scheme satisfies secrecy if any unauthorized subset of fewer than $k$ parties learns essentially no information about the secret state.
A fully formal definition is given in Appendix~\ref{def:secrecy-formal}.
\end{definition}

In a traceable secret-sharing scheme, secrecy must continue to hold even in the presence of the tracer. In particular, access to the tracing key $tk$ should not weaken the secrecy guarantee. We formalize this stronger requirement below.

\begin{definition}[Tracer secrecy (informal)]
\label{def:tsecrecy}
An LVT-QSS scheme satisfies tracer secrecy if any unauthorized subset of fewer than $k$ parties obtains only negligible information about the shared secret state, even when the tracing key $tk$ is available to it.
A fully formal definition is given in Appendix~\ref{def:tsecrecy1-formal}.
\end{definition}

A leak-verifiable and traceable quantum secret-sharing scheme must satisfy traceability. Let $\mathcal{I}_{\mathrm{idx}}\subseteq [m]$ be a coalition of leaked parties with $|\mathcal{I}_{\mathrm{idx}}|=f<k$, and let $\mathcal{I}$ denote the corresponding tag set. Using the shares in its possession, the coalition may construct a black-box reconstruction procedure $\mathcal{R}$ that takes as input an additional set of $k-f$ shares and attempts to recover the secret state. Traceability requires that whenever such a reconstruction device succeeds with sufficiently high probability, a tracing algorithm given black-box access to $\mathcal{R}$ can identify the coalition responsible for constructing $\mathcal{R}$.

To make this precise, let $Q\subseteq [m]\setminus \mathcal{I}_{\mathrm{idx}}$ be a subset of size $k-f$, chosen uniformly at random from the parties outside the coalition. If, on input the shares $\{|sh_i\rangle\}_{i\in Q}$, the box $\mathcal{R}$ outputs a state negligibly close to the original secret, then we regard $\mathcal{R}$ as a good reconstruction box. We formalize this notion below.

\begin{definition}[Good reconstruction boxes (informal)]
\label{def:good-box}
A reconstruction box constructed from the leaked shares of a coalition of $f<k$ parties is called good if, when supplemented with a uniformly random set of $k-f$ shares from parties outside the coalition, it reconstructs the secret state up to negligible trace distance with probability at least $\epsilon$.
A fully formal definition is given in Appendix~\ref{def:good-box-formal}.
\end{definition}

\begin{definition}[Traceability (informal)]
\label{trace1}
An LVT-QSS scheme is traceable if any useful reconstruction box constructed from leaked shares can be attributed to the coalition responsible for those shares; in particular, if the box reconstructs the secret using additional shares from parties outside the coalition, then the tracing procedure identifies the coalition's tag set except with negligible probability. The full tracing experiment and formal definition are given in Appendix~\ref{def:trace1-formal}.
\end{definition}

\begin{definition}[Non-imputability (informal)]
\label{def:ni}
An LVT-QSS scheme is non-imputable if accepted tracing evidence cannot falsely identify an honest participant as responsible for share leakage, except with negligible probability. The full non-imputability experiment and formal definition are given in Appendix~\ref{def:ni-formal}.
\end{definition}

\subsubsection{Leakage-verifiable and traceable quantum secret sharing protocol}
\label{subsub:leakagepro}
The construction relies on standard cryptographic primitives. In particular, we use a quantum-secure collision-resistant and preimage-resistant hash function $H$ and a post-quantum IND-CPA-secure symmetric-key encryption scheme $\Pi_{\mathsf{SKE}}=(\mathrm{Enc},\mathrm{Dec}).$
We further assume that $\Pi_{\mathsf{SKE}}$ satisfies the sparse-hit property with parameter $\alpha_{\max}(\lambda)$. For a key $k\in\{0,1\}^{\lambda}$ and a message $m\in\{0,1\}^*$, the encryption algorithm outputs a ciphertext
$c\leftarrow \mathrm{Enc}(m,k)\in\{0,1\}^{\ell(|m|)},$ whereas the decryption algorithm $\mathrm{Dec}$, on input $(c,k)$, outputs either $m$ or $\bot$.

For any QPT adversary $\mathcal{A}$, let $\mathrm{Adv}^{\mathrm{IND\text{-}CPA}}_{\mathcal{A}}(\lambda)$ denote the advantage of $\mathcal{A}$ in the post-quantum IND-CPA experiment against $\Pi_{\mathsf{SKE}}$. We assume that this advantage is negligible in the security parameter $\lambda$. We now describe the constituent algorithms in detail.

The dealer holds an arbitrary secret quantum state $|\psi\rangle = \sum_s \alpha_s |s\rangle \in \mathcal{H}_\lambda.$
The goal of the LVT-QSS scheme is to securely share this quantum state among \(m\) parties, such that any subset of at least \(k\) parties can later reconstruct the state. The process of sharing the quantum state involves two main steps: setting up the system parameters and generating the secret shares. These algorithms ensure that the shares are distributed securely and that the secret state can be reconstructed only by authorized parties.

\protocol{Initialization and share generation of LVT-QSS}
\textbf{1. Setup}$(1^\lambda,m)$.
On input the security parameter \(1^\lambda\) and the number of parties \(m\):
\begin{itemize}[leftmargin=1.2em,nosep]
    \item Sample distinct tags $tag_0,\ldots,tag_{m-1}\xleftarrow{\$}\mathbb{F},$
    and set $Tags := \{tag_0,\ldots,tag_{m-1}\}.$
    
    \item Sample a tracing key $tk \xleftarrow{\$} \{0,1\}^{\lambda}.$
\end{itemize}
\textbf{2. GenSecretShares}$(|\psi\rangle,m,k,Tags,tk)$.
On input a secret state \(|\psi\rangle\), parameters \(m\) and \(k\), the tag set \(Tags\), and the tracing key \(tk\):
\begin{itemize}[leftmargin=1.2em,nosep]
    \item Let \(\mathbb{F}=\mathbb{Z}_q\), where \(q\) is prime and \(\lceil \log_2 q\rceil=\lambda\). Sample \(m\) distinct elements $x_0,\ldots,x_{m-1}\in\mathbb{F},$
    and a random key $uk\xleftarrow{\$}\mathbb{F}.$

    \item For each \(i\in\{0,\ldots,m-1\}\), compute
    \[
    ct_i \leftarrow \mathrm{Enc}(tag_i,tk),\
    vk_i \leftarrow H(tag_i),\ \text{and}\
    r_i \leftarrow \mathrm{Enc}(ct_i,uk).
    \]

    \item For any vector \(c=(c_0,\ldots,c_{k-1})\in\mathbb{F}^k\), define $p_c(t)=c_0+c_1 t+\cdots+c_{k-1} t^{k-1}.$

    \item Encode the secret state as
    \[
    \begin{aligned}
    |\psi\rangle \mapsto\;
    &\sum_s \alpha_s
    \sum_{\substack{c\in\mathbb{F}^k\\ c_0=uk,\,c_{k-1}=s}}
    \\
    &\quad |p_c(x_0),r_0\rangle \otimes \cdots \otimes |p_c(x_{m-1}),r_{m-1}\rangle .
    \end{aligned}
    \]

    \item Initialize registers \((R_0,A_0),\ldots,(R_{m-1},A_{m-1})\) according to the above state. Denote the quantum share in register pair \((R_i,A_i)\) by \(|sh_i\rangle\) for each \(i\in\{0,\ldots,m-1\}\).
\end{itemize}
\noindent\rule{\linewidth}{0.5pt}

We now describe the recovery algorithm. Given any authorized set containing at least \(k\) valid shares, the algorithm reconstructs the secret quantum state. The recovery process involves selecting \(k\) valid shares from the set, applying a Vandermonde transform to reconstruct the secret, and removing any residual dependencies to fully recover the state.

\protocol{Secret recovery in LVT-QSS}
\textbf{3. SecretRecover}$(\{|sh_i\rangle\}_{i\in P})$.
On input a set of shares indexed by $P\subseteq [m]$:
\begin{itemize}[leftmargin=1.2em, nosep]
    \item If $|P|<k$, output $\perp$. Otherwise, choose a subset $P'\subseteq P$ such that $|P'|=k$. Without loss of generality, relabel the corresponding register pairs as
    $(R_0,A_0),\ldots,(R_{k-1},A_{k-1}).$

    \item Apply the inverse Vandermonde transform $V_k(x_0,\ldots,x_{k-1})^{-1}$ to registers $R_0,\ldots,R_{k-1}$.

    \item Cyclically right-shift the $k$ register pairs $(R_0,A_0),\ldots,(R_{k-1},A_{k-1})$ by one position.

    \item Copy $R_1$ into an ancilla register $B$ via $|uk\rangle_{R_1}|0\rangle_B \mapsto |uk\rangle_{R_1}|uk\rangle_B.$

    \item Apply the linear transform $V_{k-1}(x_k,\ldots,x_{m-1})$ to registers $R_1,\ldots,R_{k-1}$.

    \item For each $i=1,\ldots,k-1$, remove the residual dependence by adding $R_0\cdot (x_{k+i-1})^{k-1}$ to register $R_i$.

    \item Output the content of register $R_0$ as the recovered secret state.
\end{itemize}
\noindent\rule{\linewidth}{0.5pt}

We next describe the tracing and verification algorithms. First, the tracing algorithm takes as input the tracing key and a set of valid leaked shares, and extracts the identities of the leaking parties by tracing the shares back to their corresponding parties. Then, the verification algorithm takes the verification keys and a candidate identity set \(\mathcal{I}'\) returned by the tracing algorithm, and verifies whether the traced identities are valid by checking them against the verification keys.

\protocol{Secret tracing and verification in LVT-QSS}
\textbf{4. SecretTrace}$(tk,\{|sh_i\rangle\}_{i \in Q})$.
On input a tracing key $tk$ and a collection of valid shares indexed by a set $Q \subseteq [m]\setminus \mathcal{I}_{\mathrm{idx}}$ with $|Q|=k-f:$
\begin{itemize}[leftmargin=1.2em, nosep]
    \item For each $i \in Q$, copy the content of $A_i$ into an ancilla register $C_i$: 
    $|r_i\rangle_{A_i}|0\rangle_{C_i} \mapsto |r_i\rangle_{A_i}|r_i\rangle_{C_i}$.
        
    \item Feed the register pairs $\{(R_i,A_i)\}_{i\in Q}$ into the reconstruction box $\mathcal{R}$, which outputs an ancilla register $B$ and auxiliary registers $\tilde A_0,\ldots,\tilde A_{k-1}.$
 
    \item Measure $B$ to obtain $uk$, and measure $\tilde A_0,\ldots,\tilde A_{k-1}$ to obtain $\{r_i\}_{i=0}^{k-1}$. Set $I \leftarrow \{ r_i \}_{i=0}^{k-1}$.
        
    \item Measure each $C_i$ to obtain $r'_i$ and set $I'' \leftarrow \{ r'_i \}_{i \in Q}$.
        
    \item Update $I \leftarrow I \setminus I''$ and output
    \[
    \mathcal{I}' \leftarrow \{ \mathrm{Dec}(\mathrm{Dec}(r,uk),tk) : r \in I\}.
    \]
\end{itemize}  
\textbf{5. SecretVerify}$(\{vk_i\}_{i=0}^{m-1}, \mathcal{I}')$.
On input the verification keys $\{vk_i\}_{i=0}^{m-1}$ and a candidate evidence set $\mathcal{I}'$:
\begin{itemize}[leftmargin=1.2em, nosep]
    \item If $|\mathcal{I}'| \neq f$, output $\perp$.
        
    \item Enumerate $\mathcal{I}' = \{tag_1,\ldots,tag_f\}$. For each $tag \in \mathcal{I}'$, compute $v := H(tag)$. If $v \notin \{vk_i\}_{i=0}^{m-1}$, output $\perp$.
        
    \item Otherwise, output 1.
\end{itemize}    
\noindent\rule{\linewidth}{0.5pt}

\subsubsection{Proof}
\label{subsub:proofLVT}

\begin{theorem}
\label{th:correct1}
The $(k,m)$-threshold LVT-QSS scheme is correct.
\end{theorem}

\begin{proof}[Proof sketch]
Correctness follows from the linear structure of the coherent threshold encoding. The shares are evaluations of a degree-\((k-1)\) polynomial over distinct field points, with the secret embedded in one coefficient. Any authorized set of at least \(k\) shares contains \(k\) evaluations that define an invertible Vandermonde system, allowing the reconstruction procedure to recover the secret coefficient by a reversible linear transformation. Hence the recovered state is negligibly close to the original secret state except with negligible probability. The complete proof is given in Appendix~\ref{app:lvt-correctness-proof}.
\end{proof}

\begin{theorem}
\label{th:secrecy}
The $(k,m)$-threshold LVT-QSS scheme satisfies secrecy.
\end{theorem}

\begin{proof}[Proof sketch]
Secrecy follows from the random degrees of freedom in the coherent polynomial encoding. 
For a fixed secret basis value \(s\), the coefficient \(c_{k-1}=s\) is fixed, while the remaining coefficients, including \(c_0=uk\), are sampled independently of \(s\). 
Thus, any unauthorized set \(M\) with \(|M|<k\) observes fewer than \(k\) evaluations of a degree-\((k-1)\) polynomial whose remaining \(k-1\) coefficients are uniformly random. 
The induced distribution of the evaluations seen by \(M\) is therefore uniform and independent of \(s\). 
After tracing out the purification environment and the registers outside \(M\), the reduced state on the adversary's view is independent of the secret state. 
Hence any measurement on fewer than \(k\) shares has only negligible distinguishing advantage. 
The complete proof is given in Appendix~\ref{app:lvt-secrecy-proof}.
\end{proof}

\begin{theorem}
\label{th:tsecrecy}
The $(k,m)$-threshold LVT-QSS scheme satisfies tracer secrecy.
\end{theorem}

\begin{proof}[Proof sketch]
The share registers \(R_M\) of any unauthorized subset \(M\), with \(|M|<k\), are already independent of the secret by secrecy of the underlying threshold sharing. The auxiliary registers \(A_M\) contain only tracing-related data generated from the tracing key, local randomness, and the user key, all of which are independent of the secret state. Therefore, even if the tracing key is included as a classical register \(T\), the joint state on \(\mathrm{View}_M\otimes T\) is independent of which secret was shared. Hence the distinguishing advantage of any unauthorized subset is negligible. The complete proof is given in Appendix~\ref{app:lvt-tracer-secrecy-proof}.
\end{proof}

\begin{theorem}[Informal]
\label{th:lvt-trace}
Assume that \(H:\mathbb{F}\to\mathbb{F}\) is quantum-secure collision-resistant and preimage-resistant, and that the symmetric-key encryption scheme
\(\Pi_{\mathrm{SKE}}=(\mathrm{Enc},\mathrm{Dec})\) satisfies the sparse-hit property with negligible parameter \(\alpha_{\max}(\lambda)\).
If the number of sparse-hit trials is polynomially bounded in \(\lambda\), then the LVT-QSS scheme is traceable against any QPT adversary: every useful reconstruction box constructed from leaked shares can be traced to the coalition that created it, except with negligible probability.
\end{theorem}

\begin{proof}[Proof sketch]
A successful tracing attack must produce a good reconstruction box while causing either a valid tracing output to be rejected or an incorrect tag set to be accepted. The first case is ruled out by the completeness of the verification procedure. In the second case, the accepted tag set either contains a forged tag, which yields a collision under \(H\), or contains an honest tag accepted through a sparse-hit event in the encryption layer. Collision resistance and the negligible sparse-hit probability bound these cases, and a union bound over polynomially many trials gives negligible tracing advantage. The complete proof is given in Appendix~\ref{app:lvt-trace-proof}.
\end{proof}

\begin{theorem}[Informal]
\label{th:lvt-ni}
Assume that \(H:\mathbb{F}\to\mathbb{F}\) is quantum-secure collision-resistant and preimage-resistant. Then the LVT-QSS scheme is non-imputable against any QPT adversary: no honest participant can be falsely accused of being a leaker, except with negligible probability.
\end{theorem}

\begin{proof}[Proof sketch]
The proof reduces any successful non-imputability attack to either finding a preimage or a collision for \(H\). 
We replace the verification key of the challenged honest party by an independently sampled hash value \(y^*=H(w^*)\), which is identically distributed to the real key. 
If an adversary nevertheless causes this honest party to be falsely implicated, then the accepted tag must either hash to \(y^*\), giving a preimage of the challenge, or hash to the verification key of another party, giving a collision in \(H\). 
Since \(H\) is quantum-secure preimage resistant and collision resistant, both events occur only with negligible probability. 
The complete proof is given in Appendix~\ref{app:lvt-ni-proof}.
\end{proof}

\subsection{Malicious-but-accountable model}
\label{sub:malicious}

We now consider a malicious-but-accountable reconstruction model. Let $V:=\{tag_i \mid i\in V_{\mathrm{idx}}\}$ and $U:=\{tag_i \mid i\in U_{\mathrm{idx}}\}$ denote the malicious and honest tag sets, respectively. Here $V_{\mathrm{idx}},U_{\mathrm{idx}}\subseteq [m]$ are the corresponding index sets, with $|V_{\mathrm{idx}}|=t<k$. Unlike in the honest-reconstruction setting, malicious parties may submit forged shares, and hence the reconstruction-box abstraction is insufficient. Instead, reconstruction must be modeled explicitly, and tracing is performed with respect to the auxiliary information generated during the reconstruction procedure.

The condition $t<k$ is essential: if $t\ge k$, then the malicious parties can reconstruct and leak the secret using only their own shares, so no nontrivial tracing guarantee can be required in general. We therefore restrict attention to the regime $t<k$.

In this model, the role of tracing is to provide accountability. More precisely, the tracer produces publicly verifiable compliance proofs for parties that behaved honestly during reconstruction, thereby excluding them as suspects and isolating the malicious parties responsible for submitting forged shares. The corresponding notions of traceability and non-imputability are formalized in Sec.~\ref{subsub:secdefmalicious}.

\subsubsection{Security definitions}
\label{subsub:secdefmalicious}

We define an MTA-QSS scheme for securely sharing a quantum state $|\psi\rangle=\sum_s \alpha_s |s\rangle$ among $m$ parties in the malicious-but-accountable setting. The scheme is required to support secret reconstruction from any subset of at least $k$ shares, together with tracing and public verification based on the auxiliary information generated during reconstruction.

Formally, an MTA-QSS scheme is specified by the following five QPT algorithms.
\begin{itemize}
    \item \textbf{Setup}: Given the security parameter $1^\lambda$ and the number of parties $m$, the trusted authority generates a tracing key set $\{tk_i\}_{i=0}^{m-1}$ and a tag set $Tags=\{tag_i\}_{i=0}^{m-1}$.

    \item \textbf{GenSecretShares}: Given a secret quantum state $|\psi\rangle$, the number of parties $m$, the threshold parameter $k$, the tag set $Tags$, and the tracing keys $\{tk_i\}_{i=0}^{m-1}$, the dealer generates the quantum shares $\{|sh_i\rangle\}_{i=0}^{m-1}$ together with the verification keys $\{vk_i\}_{i=0}^{m-1}$, the auxiliary key $uk$, and the check key $ck$.

    \item \textbf{SecretRecover}: Given the shares indexed by a set $P\subseteq [m]$, the algorithm outputs either a reconstructed quantum state together with a tracing transcript $I$, or the failure symbol $\bot$.

    \item \textbf{SecretTrace}: Given the tracing keys $\{tk_i\}_{i=0}^{m-1}$, the auxiliary key $uk$, and the tracing transcript $I$, the tracing algorithm outputs two sets $U'$ and $V'$.

    \item \textbf{SecretVerify}: Given the verification keys $\{vk_i\}_{i=0}^{m-1}$, the auxiliary key $uk$, the check key $ck$, the sets $U'$ and $V'$, and the tracing transcript $I$, the verification algorithm outputs a bit $b\in\{0,1\}$, where $b=1$ if and only if $(U',V')$ satisfies the scheme's verification condition.
\end{itemize}

Correctness for an MTA-QSS scheme is defined in the usual threshold sense: any subset of at least $k$ valid shares reconstructs the secret.

\begin{definition}[Correctness (informal)]
\label{def:correct2}
An MTA-QSS scheme is correct if any subset $P\subseteq [m]$ with $|P|\ge k$, using honestly generated shares, can reconstruct a state that is within trace distance at most $\mathrm{negl}(\lambda)$ of the original secret state with probability at least $1-\mathrm{negl}(\lambda)$. No correctness guarantee is required when some of the input shares are forged or adversarially corrupted. A fully formal definition is given in Appendix~\ref{def:correct2-formal}.
\end{definition}

An MTA-QSS scheme is required to satisfy three security properties: secrecy, traceability, and non-imputability. Among them, secrecy is the basic confidentiality requirement, namely that any set of fewer than $k$ parties obtains no information about the secret. In the malicious-but-accountable setting, we adopt the general secrecy notion of Definition~\ref{def:secrecy}.

\begin{definition}[Tracer secrecy (informal)]
\label{def:tsecrecy2}
An MTA-QSS scheme satisfies tracer secrecy if any subset $M\subseteq [m]$ with $|M|<k$ learns essentially no information about the secret state, even when given access to the tracing keys and the auxiliary key $uk$.  
The fully formal definition is given in Appendix~\ref{def:tsecrecy2-formal}.
\end{definition}

\begin{definition}[Traceability (informal)]
\label{def:trace2}
An MTA-QSS scheme is traceable if the honest parties contributing valid shares can be identified during reconstruction, from which the parties who submitted forged shares can be inferred.  
The full tracing experiment and formal definition are given in Appendix~\ref{def:trace2-formal}.
\end{definition}

\begin{definition}[Non-imputability (informal)]
\label{def:ni2}
An MTA-QSS scheme is non-imputable if the evidence generated during reconstruction cannot falsely attribute a forged share to any honest party contributing a valid share.
The full non-imputability experiment and formal definition are given in Appendix~\ref{def:ni2-formal}.
\end{definition}

\subsubsection{Malicious-traceable and accountable quantum secret sharing protocol}
\label{subsub:maliciouspro}

We next present the malicious-traceable and accountable quantum secret sharing protocol. As in the LVT-QSS construction, it is built from the same underlying cryptographic primitives and relies on the same security assumptions. The protocol consists of five phases: setup, share generation, recovery, tracing, and verification.

The dealer holds an arbitrary secret quantum state $|\psi\rangle = \sum_s \alpha_s |s\rangle \in \mathcal{H}_\lambda.$
We begin by describing the initialization and share-generation phases of the MTA-QSS scheme. Given a secret quantum state, these procedures generate the system parameters and encode the secret into \(m\) distributed quantum shares. The recovery, tracing, and verification phases will be presented subsequently.

\protocol{Initialization and share generation of MTA-QSS}
\textbf{1. Setup}$(1^\lambda,m)$.
On input the security parameter $1^\lambda$ and the number of parties $m$, the setup proceeds as follows:
\begin{itemize}[leftmargin=1.2em,nosep]
    \item Sample $m$ distinct tags $tag_0,\ldots,tag_{m-1}\xleftarrow{\$}\mathbb{F}$ and set
    \[
    Tags := \{tag_0,\ldots,tag_{m-1}\}.
    \]
    \item Choose $m$ pairwise distinct field elements $g_0,\dots,g_{m-1}\in\mathbb{F}$, and compute
    \[
    tk_i := H(g_i), \qquad i=0,\ldots,m-1,
    \]
    where each $tk_i$ is encoded as a $\lambda$-bit string with $\lambda=\lceil \log_2 q\rceil$.
\end{itemize}
\textbf{2. GenSecretShares}$(|\psi\rangle,m,k,Tags,\{tk_i\}_{i=0}^{m-1})$.
On input a secret state $|\psi\rangle$, the parameters $m$ and $k$, the tag set $Tags$, and the tracing key set $\{tk_i\}_{i=0}^{m-1}$, proceed as follows:
\begin{itemize}[leftmargin=1.2em,nosep]
    \item Select $m$ distinct elements $x_0,\dots,x_{m-1}\in\mathbb{F}$.
    \item Sample an auxiliary key $uk\xleftarrow{\$}\mathbb{F}$.
    \item Compute the check key $ck := H\!\left(\bigoplus_{i=0}^{m-1} tk_i\right).$
    \item For each $i\in\{0,\dots,m-1\}$, compute
    \[
    ct_i \leftarrow \mathrm{Enc}(tag_i,ck),\
    vk_i \leftarrow H(tag_i),\text{ and }
    r_i \leftarrow \mathrm{Enc}(ct_i,uk).
    \]
    \item Let $\mathbb{F}=\mathbb{Z}_q$. For any vector $c=(c_0,\dots,c_{k-1})\in\mathbb{F}^k$, define the polynomial $p_c(t)=c_0+c_1 t+\cdots+c_{k-1} t^{k-1}.$
    \item Encode the secret state as
    \[
    |\psi\rangle \mapsto
    \sum_s \alpha_s \sum_{\substack{c\in\mathbb{F}^k\\ c_0=uk,\,c_{k-1}=s}}
    \bigotimes_{i=0}^{m-1} |p_c(x_i), r_i\|tk_i\rangle.
    \]
    \item Initialize register pairs $(R_i,A_i)$ so that each pair contains the corresponding share and auxiliary information. Denote the quantum share in the $i$th register pair by $|sh_i\rangle$.
\end{itemize}
\noindent\rule{\linewidth}{0.5pt}

We next present the secret-recovery phase of the MTA-QSS scheme. For any subset of at least \(k\) shares, the procedure reconstructs the secret quantum state and simultaneously extracts the auxiliary classical information \(I\). This auxiliary output is subsequently used for tracing and verification.

\protocol{Secret recovery in MTA-QSS}
\textbf{3. SecretRecover}$(\{|sh_i\rangle\}_{i\in P})$.
On input a set of shares indexed by $P\subseteq [m]$, proceed as follows:
\begin{itemize}[leftmargin=1.2em, nosep]
    \item If $|P|<k$, output $\perp$. Otherwise, choose any subset $P'\subseteq P$ of size $k$ and relabel the corresponding register pairs as $(R_0,A_0),\dots,(R_{k-1},A_{k-1}).$

    \item Apply the inverse Vandermonde transform $V_k(x_0,\dots,x_{k-1})^{-1}$ to the registers $R_0,\dots,R_{k-1}$.
    \item Cyclically right-shift the $k$ register pairs to obtain
    \[
    ((R_{k-1},A_{k-1}),(R_0,A_0),\dots,(R_{k-2},A_{k-2})).
    \]
    \item Apply the linear transform $V_{k-1}(x_k,\dots,x_{m-1})$ to the registers $R_1,\dots,R_{k-1}$.
    \item Remove the residual dependence by adding $R_0\cdot (x_{k+i-1})^{k-1}$ to $R_i$ for each $i=1,\dots,k-1$.
    \item Measure the ancillary registers $A_0,\dots,A_{k-1}$ to obtain the classical outcomes $r_i\|tk_i$, and define $I:=\{\,r_j\|tk_j\,\}_{j=0}^{k-1}.$
    \item Output the content of $R_0$ as the recovered secret state.
\end{itemize}
\noindent\rule{\linewidth}{0.5pt}

We next describe the tracing and verification phases of the MTA-QSS scheme. Given the auxiliary classical information \(I\) obtained during recovery, the tracing procedure extracts the associated tags and partitions them according to the tracing key. The verification procedure then checks that the recovered tags are consistent with the verification keys, thereby validating the traced evidence.

\protocol{Secret tracing and verification in MTA-QSS}
\textbf{4. SecretTrace}$(\{tk_i\}_{i=0}^{m-1}, uk, I)$.
On input a tracing key set $\{tk_i\}_{i=0}^{m-1}$, the auxiliary key $uk$, and a classical information set $I$, proceed as follows:
\begin{itemize}[leftmargin=1.2em, nosep]
    \item Compute the check key $ck := H(\bigoplus_{i=0}^{m-1} tk_i).$

    \item For each $I_j\in\mathcal{I}$, let $I_j^{(\lambda)}$ denote the last $\lambda$ bits of $I_j$. Partition $I$ into
    \[
    \mathcal{U} := \{\, I_j \mid I_j^{(\lambda)} \in \{tk_i\}_{i=0}^{m-1} \,\}
    \]
    and
    \[
    \mathcal{V} := \{\, I_j \mid I_j^{(\lambda)} \notin \{tk_i\}_{i=0}^{m-1} \,\}.
    \]

    \item For each $I_j\in \mathcal{U}\cup \mathcal{V}$, let $r_j$ denote the substring of $I_j$ obtained by removing the last $\lambda$ bits, and compute
    \[
    ct_j := \mathrm{Dec}(r_j,uk),
    \qquad
    tag_j := \mathrm{Dec}(ct_j,ck).
    \]

    \item Define
    \[
    U' := \{\, tag_j \mid I_j\in \mathcal{U} \,\},
    \qquad
    V' := \{\, tag_j \mid I_j\in \mathcal{V} \,\}.
    \]

    \item Return the pair $(U',V')$.
\end{itemize}
\textbf{5. SecretVerify}$(\{vk_i\}_{i=0}^{m-1}, uk, ck, U', V', I)$.
On input the verification key set $\{vk_i\}_{i=0}^{m-1}$, the auxiliary key $uk$, the check key $ck$, the sets $U'$ and $V'$, and the classical information set $I$, proceed as follows:
\begin{itemize}[leftmargin=1.2em, nosep]
    \item Parse $I=\{I_j\}_{j=0}^{k-1}$ and extract $\{r_j\}_{j=0}^{k-1}$ by removing the last $\lambda$ bits of each $I_j$.

    \item Compute
    \[
    U'' := \{\, \mathrm{Enc}(\mathrm{Enc}(tag_z,ck),uk) \mid tag_z\in U' \,\},
    \]
    \[
    V'' := \{\, \mathrm{Enc}(\mathrm{Enc}(tag_z,ck),uk) \mid tag_z\in V' \,\}.
    \]

    \item If $U'' \cup V'' \neq \{r_j\}_{j=0}^{k-1},$
    output $\perp$.

    \item If there exists $tag_z\in V'$ such that $H(tag_z)\in \{vk_i\}_{i=0}^{m-1},$
    output $\perp$.

    \item If for all $tag_z\in U'$, $H(tag_z)\in \{vk_i\}_{i=0}^{m-1},$
    output $1$; otherwise output $\perp$.
\end{itemize}
\noindent\rule{\linewidth}{0.5pt}

\subsubsection{Proof}
\label{proof:MTA}

\begin{theorem}
\label{th:correct2}
The $(k,m)$-threshold MTA-QSS scheme is correct.
\end{theorem}

\begin{proof}[Proof sketch]
The MTA-QSS scheme uses the same coherent $(k,m)$-threshold encoding and decoding maps on the quantum share registers as the LVT-QSS scheme. The only difference is that MTA-QSS appends additional auxiliary classical information, namely $(r_j\|tk_j)$ instead of $r_j$, to each share. These auxiliary registers are not used by the decoding map on the quantum share registers and therefore do not affect the recovered secret state. Correctness then follows from the correctness of LVT-QSS in Theorem~\ref{th:correct1}. A full proof is given in Appendix~\ref{app:mta-correctness-proof}.
\end{proof}

\begin{theorem}
\label{th:secrecy2}
The $(k,m)$-threshold MTA-QSS scheme satisfies secrecy.
\end{theorem}

\begin{proof}[Proof sketch]
For any unauthorized set $M\subset[m]$ with $|M|<k$, the MTA-QSS scheme uses the same quantum encoding on the share registers as the LVT-QSS scheme. Hence, the reduced state on the quantum share registers held by $M$ is independent of
the secret, by Theorem~\ref{th:secrecy}. The additional auxiliary registers in MTA-QSS contain only classical information generated from protocol randomness, such as tags and tracing-related data, and are independent of the secret state. Therefore, the joint view of any unauthorized set is independent of the secret. A full proof is given in Appendix~\ref{app:mta-secrecy-proof}.
\end{proof}

\begin{theorem}
\label{th:tracer-secrecy2}
The $(k,m)$-threshold MTA-QSS scheme satisfies tracer secrecy.
\end{theorem}

\begin{proof}[Proof sketch]
For any unauthorized set $M\subseteq[m]$ with $|M|<k$, the quantum share registers held by $M$ are distributed exactly as in the LVT-QSS scheme and are therefore independent of the secret. The additional information available in
the tracer-secrecy experiment consists of auxiliary classical registers, the tracing keys, and the auxiliary key $uk$. These values are generated from protocol randomness and tracing data, independently of the input secret state.
Thus, even conditioned on the tracing keys, the joint state of the adversary's view and the tracer-related registers is independent of the secret. The tracer-secrecy advantage is therefore negligible. A full proof is given in Appendix~\ref{app:mta-tracer-secrecy-proof}.
\end{proof}

\begin{theorem}[Informal]
\label{th:trace2}
The MTA-QSS scheme is traceable against any QPT adversary: during reconstruction, the honest parties contributing valid shares can be identified except with negligible probability, from which the parties who submitted forged shares can be inferred.
\end{theorem}

\begin{proof}[Proof sketch]
The proof follows the formal tracing experiment in Appendix~\ref{def:trace2-formal}. Let $U$ be the tag set corresponding to the honestly generated shares used in reconstruction, and let $(U',V')\leftarrow \mathrm{SecretTrace}(tk,uk,I)$. The adversary can win only if one of the following events occurs: the tracing evidence fails verification; verification accepts a set $U'$ containing a forged tag; or verification accepts a tag set $U'\subseteq Tags$ with $U'\neq U$.

The first event occurs only with negligible probability by the consistency of $\mathrm{SecretTrace}$ and $\mathrm{SecretVerify}$ on honestly generated reconstruction evidence. The second event yields a collision in the hash function used to derive the verification keys, and hence occurs only with
negligible probability. The third event is ruled out because the reconstruction evidence $I$, together with $uk$ and $ck$, uniquely determines the tags recovered by $\mathrm{SecretTrace}$ and checked by $\mathrm{SecretVerify}$. Therefore, every possible winning event has negligible probability, and the adversary's traceability advantage is negligible. A full
proof is given in Appendix~\ref{app:mta-tracer2-proof}.
\end{proof}

\begin{theorem}[Informal]
\label{th:ni2}
The MTA-QSS scheme is non-imputable against any QPT adversary: the evidence generated during reconstruction cannot falsely attribute a forged share to any honest party contributing a valid share, except with negligible probability.
\end{theorem}

\begin{proof}[Proof sketch]
The proof follows the non-imputability experiment in
Appendix~\ref{def:ni2-formal}. Let $i^*$ be the designated honest participant that the adversary attempts to frame, and let $tag^*$ be its tag. Since $i^*$ is honest, $tag^*$ is an honestly generated tag and its hash is contained in the verification-key set.

For the adversary to win, it must output evidence that passes
$\mathrm{SecretVerify}$ while placing $tag^*$ in the set of blamed tags. This is impossible under the verification rule: any tag accepted as belonging to the blamed set must fail the honest-tag verification condition, whereas $tag^*$
satisfies it. Hence the framing event cannot occur, and in particular occurs with negligible probability. A full proof is given in Appendix~\ref{app:mta-ni2-proof}.
\end{proof}

\subsection{Performance comparison}
\label{sub:per}

\begin{figure*}[!t]
    \centering

    \begin{subfigure}[t]{0.32\textwidth}
        \centering
        \includegraphics[width=\linewidth]{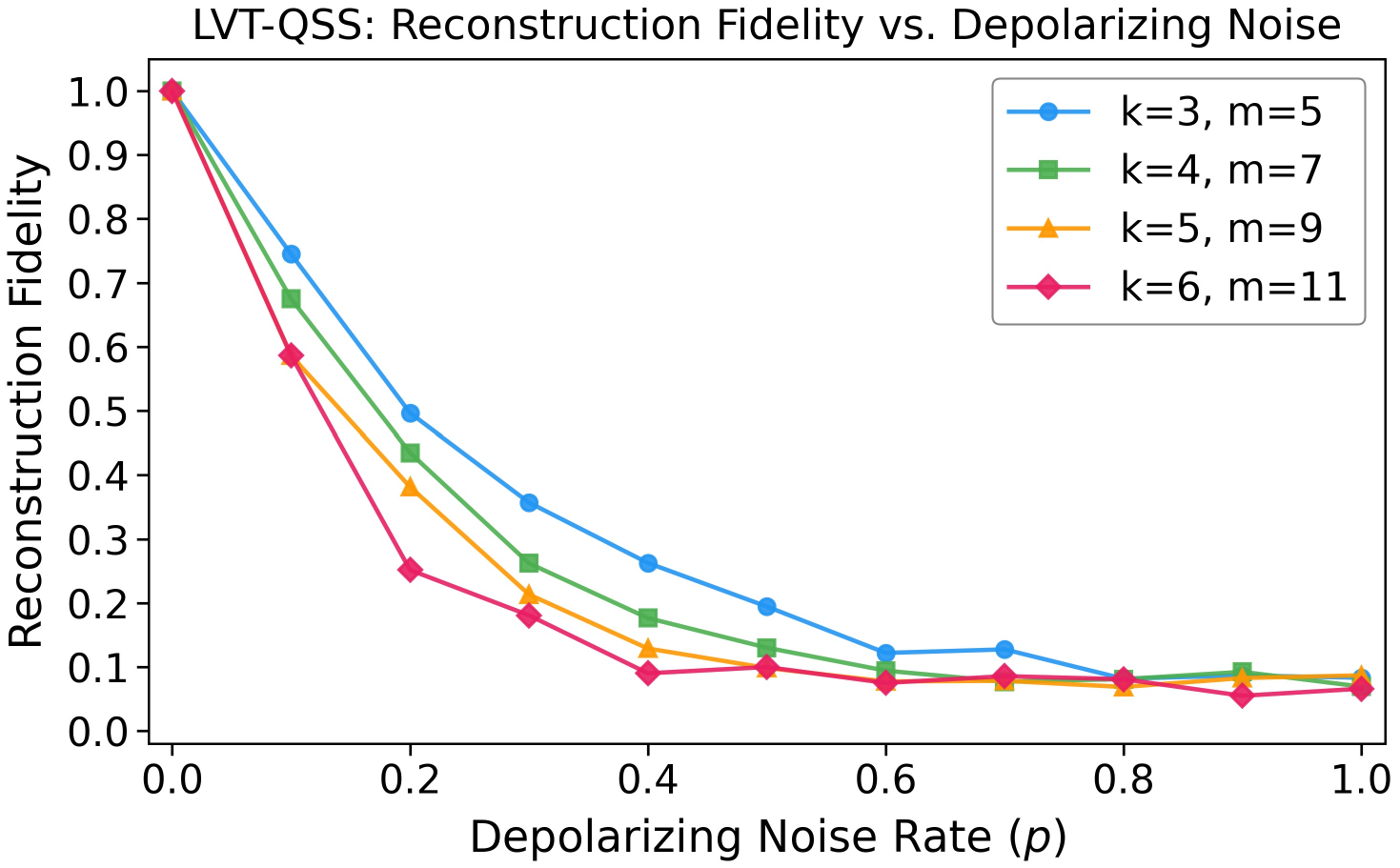}
        \caption{Reconstruction fidelity.}
        \label{fig:fidelity}
    \end{subfigure}
    \hfill
    \begin{subfigure}[t]{0.32\textwidth}
        \centering
        \includegraphics[width=\linewidth]{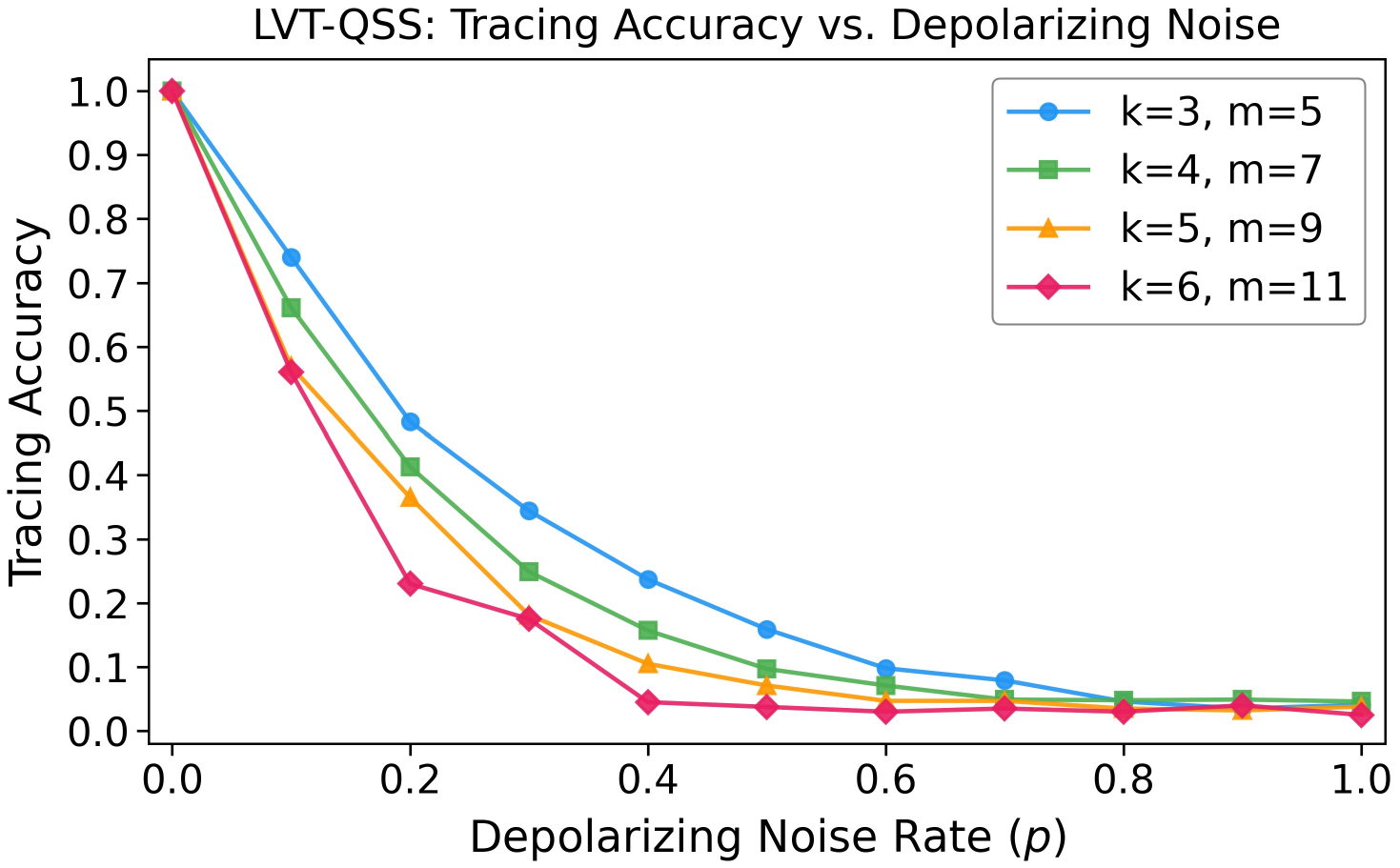}
        \caption{Tracing accuracy.}
        \label{fig:tracing}
    \end{subfigure}
    \hfill
    \begin{subfigure}[t]{0.32\textwidth}
        \centering
        \includegraphics[width=\linewidth]{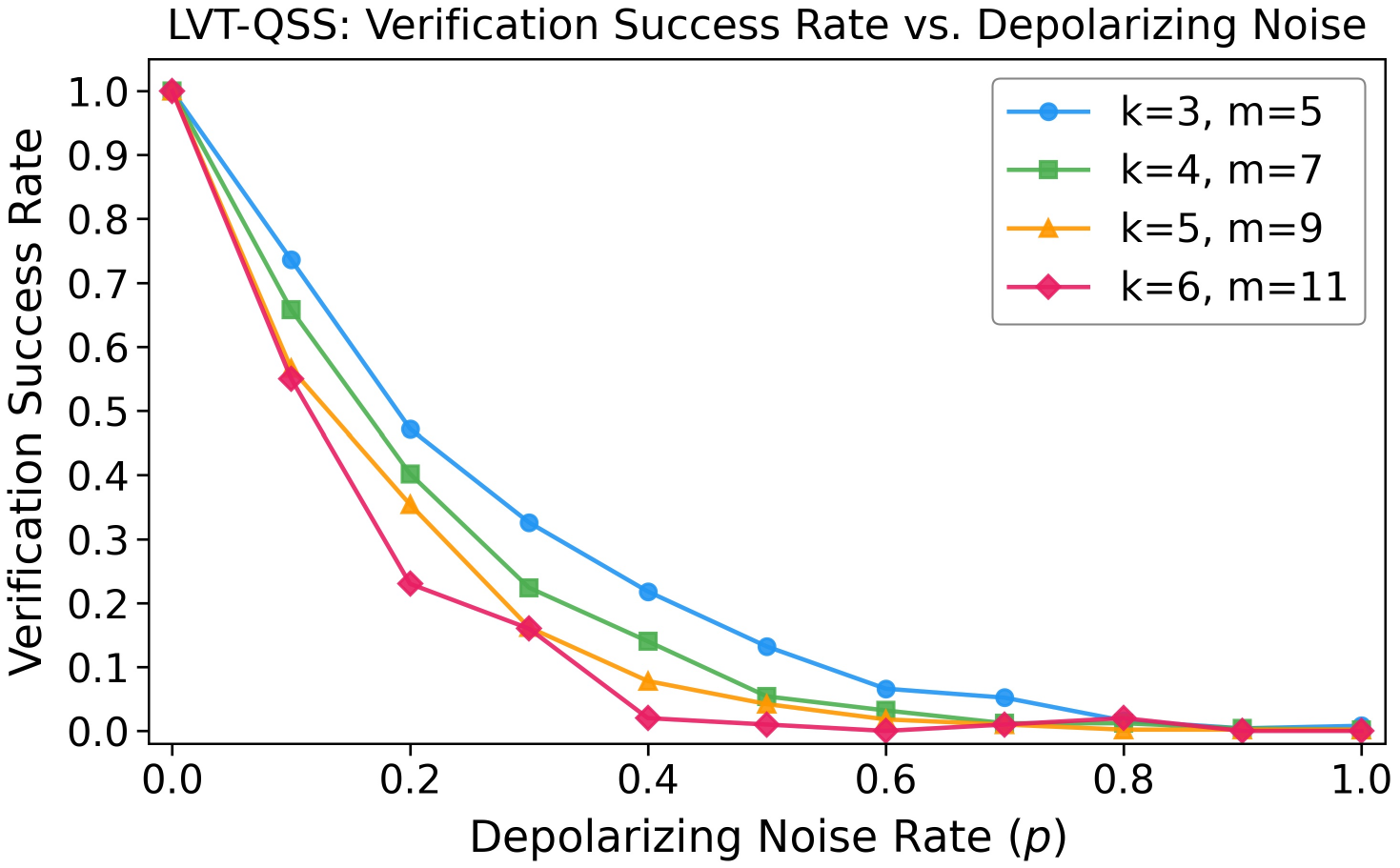}
        \caption{Verification success rate.}
        \label{fig:verification}
    \end{subfigure}

    \captionsetup{justification=raggedright,singlelinecheck=false}
    \caption{
    Noise resilience of LVT-QSS under per-share depolarizing noise for different parameter sets \((k,m)\).
    Panels show (a) reconstruction fidelity, (b) tracing accuracy, and (c) verification success rate, with \(f=2\) leakers for tracing and verification.
    }
    \label{fig:lvt_noise_comparison}
\end{figure*}

To assess the robustness of the LVT-QSS protocol under noise, we adopt a model in which each recovery share independently undergoes depolarizing noise with probability $p$. Since secret recovery in LVT-QSS relies on Lagrange interpolation, we employ a conservative approximation: if any single recovery share is corrupted, the recovery round is considered failed, and the output state is approximated as a random guess. The expected reconstruction fidelity is therefore
\[
F(p) = (1-p)^k + \frac{1-(1-p)^k}{d},
\]
where the second term represents the baseline fidelity after interpolation failure, and $d$ denotes the dimension of the secret Hilbert space $\mathcal{H}_\lambda$.

We evaluate three performance metrics for four representative parameter sets $(k,m)=(3,5),(4,7),(5,9),(6,11)$: the reconstruction fidelity $F$, the tracing accuracy (for $f=2$ leakers), and the verification success rate (for $f=2$ leakers). Each data point is averaged over 500 independent Monte Carlo trials.

As shown in Fig.~\ref{fig:lvt_noise_comparison}, all metrics decrease monotonically with increasing noise probability $p$, with smaller access structures consistently exhibiting stronger robustness. The reconstruction fidelity and tracing accuracy exhibit highly similar decay profiles, indicating that both are primarily governed by the same recovery-failure mechanism: when any share is corrupted by noise, Lagrange interpolation fails, and both reconstruction quality and tracing capability degrade simultaneously. In contrast, the verification success rate decays more rapidly because it corresponds to a stricter joint success condition: not only must all leakers' tags be traced correctly, but the corresponding results must also pass a consistency check. Overall, LVT-QSS maintains relatively high reconstruction fidelity and tracing accuracy in low to intermediate noise, whereas verification success becomes the limiting factor at an earlier stage.

\begin{figure*}[!t]
    \centering

    \begin{subfigure}[t]{0.48\textwidth}
        \centering
        \includegraphics[width=\linewidth]{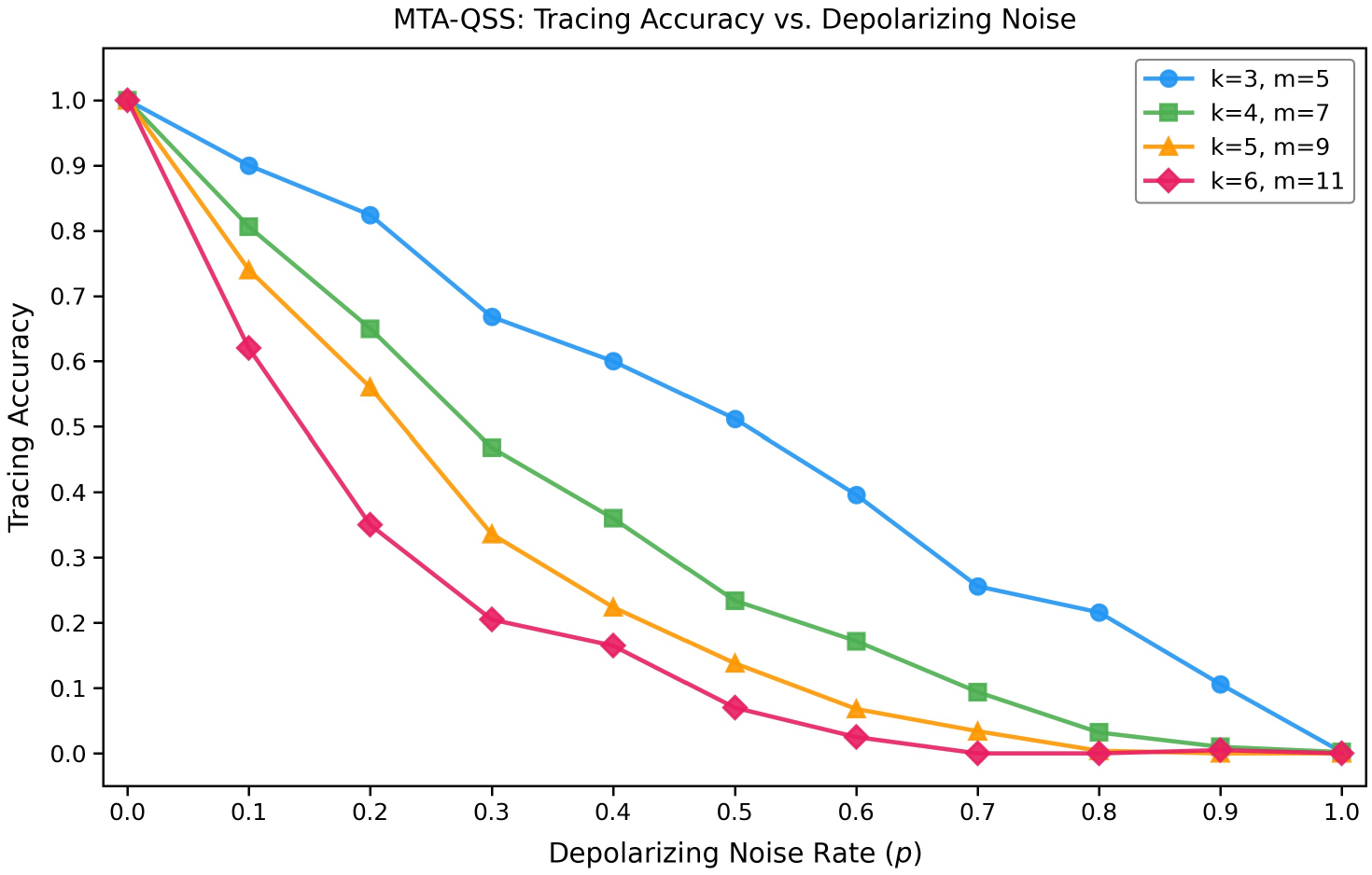}
        \caption{Tracing accuracy under depolarizing noise.}
        \label{fig:tracing2}
    \end{subfigure}
    \hfill
    \begin{subfigure}[t]{0.48\textwidth}
        \centering
        \includegraphics[width=\linewidth]{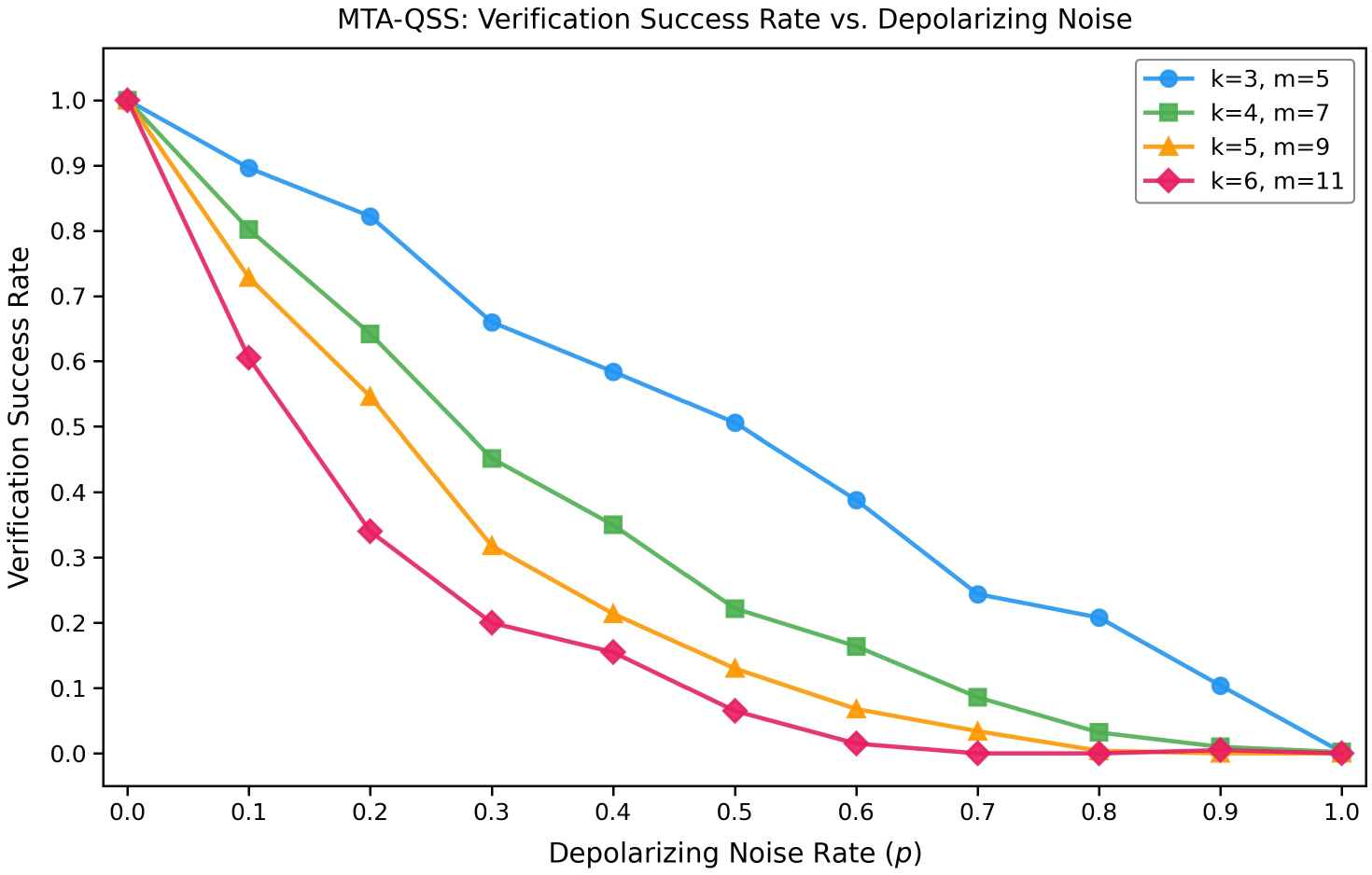}
        \caption{Verification success rate under depolarizing noise.}
        \label{fig:verification2}
    \end{subfigure}
    \captionsetup{justification=raggedright,singlelinecheck=false}
    \caption{
    Noise resilience of MTA-QSS under per-share depolarizing noise for different parameter sets \((k,m)\).
    Panels show (a) tracing accuracy and (b) verification success rate, with \(t=2\) malicious parties.
    }
    \label{fig:mta_noise_comparison}
\end{figure*}

For MTA-QSS, as shown in Fig.~\ref{fig:mta_noise_comparison}, both tracing accuracy and verification success rate decrease monotonically with noise, and smaller parameter sets remain consistently more robust, while larger sets are more sensitive. Notably, the tracing accuracy and verification success curves nearly coincide. This behavior arises from the protocol design: compared to LVT-QSS, MTA-QSS exhibits higher robustness because the auxiliary key is retained by the Dealer, so noise only affects the quantum registers during the reconstruction stage, while both tracing and verification operate solely on the classical information obtained from the reconstruction. Consequently, noise does not accumulate along the verification pathway, which effectively reduces its impact on both tracing and verification, maintaining higher accuracy under identical noise conditions. The resulting sensitivity to noise is therefore nearly identical for both metrics, yielding the observed curve coincidence.

Overall, MTA-QSS maintains high tracing accuracy and verification success in low-noise regimes. In the presence of malicious parties and local noise, the protocol performance is primarily determined by reliable tracing of honest parties' tags and the stability of subsequent consistency verification, reflecting the enhanced structural robustness of the scheme.

\section{DISCUSSION}

In this work, we move beyond the traditional focus of quantum secret sharing on secrecy and correct reconstruction, and further examine the issue of traceability in multiparty distributed settings. Existing quantum secret sharing schemes are primarily designed to prevent unauthorized access and ensure correct recovery by authorized sets; however, when valid shares are leaked or forged shares are submitted during the reconstruction stage, such schemes typically do not provide protocol level mechanisms for behavior tracing and responsibility attribution. Our results show that these requirements should not be regarded merely as external management or enforcement issues, but can instead be incorporated into the security characterization of quantum secret sharing itself.

More specifically, we show that, in quantum secret sharing, traceability corresponds to different security requirements under different types of adversarial behavior, and therefore needs to be characterized in a scenario-dependent manner. We focus on two representative settings. The first is the leakage of valid shares, where the security goal is to trace an already occurred leakage event and attribute responsibility accordingly. The second is malicious interference during the reconstruction stage, where dishonest participants attempt to disrupt secret recovery by submitting forged shares; in this setting, the protocol is required to identify the source of the malicious behavior. Based on this distinction, we introduce leakage-verifiable and traceable quantum secret sharing (LVT-QSS) and malicious-traceable and accountable quantum secret sharing (MTA-QSS), respectively, thereby extending the security definition of quantum secret sharing to include behavior tracing and responsibility attribution under abnormal situations.

At the same time, the present framework still has clear boundaries of applicability, particularly with respect to the adversarial model, implementation assumptions, and resource overhead. In the current work, we provide corresponding traceability definitions and protocol constructions separately for the two threat models of valid-share leakage and malicious forgery during reconstruction, but do not yet place them within a unified framework. An immediate question for future work is whether one can develop a unified protocol and security framework that simultaneously captures both settings, and further leads to more general, or even composable, security formulations. Such a framework would be useful for describing more realistic adversaries that may possess multiple attack capabilities at once.

It also remains important to understand whether these security properties continue to hold under stronger adversarial models, such as collusion among multiple dishonest participants, adaptive corruption, or attempts to evade tracing by exploiting device imperfections and noise. In addition, implementing traceability is likely to incur nontrivial costs, including extra classical communication, verification complexity, storage requirements, and implementation difficulty. Compared with conventional quantum secret sharing, these additional mechanisms enhance the ability to trace abnormal behavior and attribute responsibility, but may also affect fault tolerance, reconstruction efficiency, and noise robustness. A further direction worth exploring is whether one can achieve a stronger protocol guarantee in which the secret can still be correctly reconstructed even in the presence of malicious participants during the reconstruction stage, while those malicious participants can nevertheless be identified and traced. More broadly, it would also be of interest to extend the traceability perspective developed here to other distributed quantum cryptographic tasks, and to develop lower-overhead protocol variants that are more amenable to experimental implementation.

In summary, our results show that the security goals of quantum secret sharing should not be limited to secrecy and recoverability alone. In realistic multiparty settings, one must also account for the tracing of abnormal behavior and the attribution of responsibility. By introducing the formal frameworks of LVT-QSS and MTA-QSS, we provide an initial theoretical starting point for traceable quantum secret sharing, and show that incorporating traceability into the security characterization of distributed quantum cryptographic protocols is both natural and meaningful.

\begin{acknowledgments}
This work is supported by the Natural Science Foundation of Shanghai under Grant No. 25ZR1402098 and the National Natural Science Foundation of China under Grant No. 62373155.
\end{acknowledgments}

\appendix
\section{Quantum-Secure Collision-Resistant and One-Way Hash Functions}
\label{appendix:quantumhash}

Let $\lambda$ be the security parameter and $\mathbb{F} = \mathbb{F}_{q(\lambda)}$ a finite field with $\log q(\lambda) = \mathrm{poly}(\lambda)$. We consider a hash function $H: \mathbb{F} \to \mathbb{F}$.

\paragraph{Quantum Collision Resistance.}  
We say that $H$ is collision-resistant against QPT adversaries if for every QPT algorithm $\mathcal{B}$,
\begin{multline}
  \mathrm{Adv}^{\mathrm{coll}}_{\mathcal{B},H}(\lambda) = \Pr\left[ (x, x') \leftarrow \mathcal{B}(1^\lambda) : \right. \\
  \left. x \neq x' \wedge H(x) = H(x') \right] \leq \mathrm{negl}(\lambda).
\end{multline}

\paragraph{Quantum One-Wayness (Preimage Resistance).}  
We say that $H$ is one-way (preimage resistant) against QPT adversaries if for every QPT algorithm $\mathcal{B}$,
\begin{multline}
  \mathrm{Adv}^{\mathrm{ow}}_{\mathcal{B}, H}(\lambda) = \Pr\bigl[ x \leftarrow \mathcal{X}; \\
  x' \leftarrow \mathcal{B}(1^\lambda, H(x)) : H(x') = H(x) \bigr] \leq \mathrm{negl}(\lambda).
\end{multline}

Throughout this work, we assume the existence of a hash function $H$ that is both quantum-secure collision-resistant and quantum-secure one-way against QPT adversaries. These assumptions are standard in post-quantum cryptography. Lattice-based constructions, such as those by Peikert–Rosen~\cite{23_peikert2006efficient} and Lyubashevsky–Micciancio~\cite{24_lyubashevsky2006generalized}, provide concrete candidates for collision-resistant hashing. In our protocols, these hash-function properties support verification and traceability within a threshold secret-sharing framework, whose classical foundations were established by Shamir~\cite{1_rivest1979share} and Blakley~\cite{2_blakley1979safeguarding}. In particular, quantum-secure one-wayness, which also underlies constructions of quantum public-key encryption~\cite{21_kitagawa_quantum_2024,22_barooti_simple_2023}, protects participant tags against recovery from their public hash values.

\section{Formal Definitions of LVT-QSS}
\label{appendix:def-lvt-qss}

\subsection{Correctness}
\label{def:correct1-formal}
\begin{definition}[Correctness (formal)]
An LVT-QSS scheme is correct if there exists a negligible function
$\mathrm{negl}(\cdot)$ such that, for security parameter
$\lambda\in\mathbb{N}$, number of parties $m\in\mathbb{N}$, threshold
$k$ with $1\le k\le m$, subset $P\subseteq [m]$ with $|P|\ge k$, and
secret state $|\psi\rangle\in\mathcal{H}_\lambda$, the following holds.

Let $\rho:=|\psi\rangle\langle\psi|$. Generate the shares and run the
reconstruction algorithm as
\begin{align}
(tk,Tags)
&\leftarrow \mathrm{Setup}(1^\lambda,m), \nonumber\\
(sh,vk)
&\leftarrow \mathrm{GenSecretShares}(|\psi\rangle,m,k,Tags,tk), \nonumber\\
out
&\leftarrow \mathrm{SecretRecover}(\{|sh_i\rangle\}_{i\in P}).
\end{align}
where $sh := \{|sh_i\rangle\}_{i=0}^{m-1}$, and $vk := \{vk_i\}_{i=0}^{m-1}$.
Then
\begin{equation}
\Pr\!\left[
    out=\rho'
    \ \text{and}\
    \frac{1}{2}\|\rho'-\rho\|_1 \leq \mathrm{negl}(\lambda)
\right]
\geq 1-\mathrm{negl}(\lambda),
\end{equation}
where the probability is taken over the internal randomness of
$\mathrm{Setup}$, $\mathrm{GenSecretShares}$, and $\mathrm{SecretRecover}$.
\end{definition}

\subsection{Secrecy}
\label{def:secrecy-formal}

\begin{definition}[Secrecy (formal)]
\label{def:formal-secrecy}
Let $\mathcal{H}=\{\mathcal{H}_\lambda\}_{\lambda\in\mathbb{N}}$ denote the secret-state spaces of an LVT-QSS scheme.
For every $M\subseteq[m]$ with $|M|<k$ and any two secret states $|\psi_0\rangle,|\psi_1\rangle\in\mathcal{H}_\lambda$, define
\[
\mathrm{View}_M := \bigotimes_{i\in M}(R_i\otimes A_i)
\]
and let $\rho^{(b)}_{\mathrm{View}_M}$ be the reduced state generated by sharing $|\psi_b\rangle$, for $b\in\{0,1\}$.
Then there exists a negligible function $\mathrm{negl}(\cdot)$ such that, for security parameter $\lambda\in\mathbb{N}$,
\[
\frac{1}{2}
\left\|
    \rho^{(0)}_{\mathrm{View}_M}
    -
    \rho^{(1)}_{\mathrm{View}_M}
\right\|_1
\leq \mathrm{negl}(\lambda).
\]
\end{definition}

\subsection{Tracer secrecy}
\label{def:tsecrecy1-formal}

\begin{definition}[Tracer secrecy (formal)]
Let $\mathcal{H}=\{\mathcal{H}_\lambda\}_{\lambda\in\mathbb{N}}$ denote the secret-state spaces of an LVT-QSS scheme.
For every $M\subseteq[m]$ with $|M|<k$ and any two secret states $|\psi_0\rangle,|\psi_1\rangle\in\mathcal{H}_\lambda$, define
\[
\mathrm{View}_M := \bigotimes_{i\in M}(R_i\otimes A_i),
\quad
T := tk
\]
and let $\rho^{(b)}_{\mathrm{View}_M,T}$ be the reduced state generated by sharing $|\psi_b\rangle$, for $b\in\{0,1\}$. Then there exists a negligible function $\mathrm{negl}(\cdot)$ such that, for security parameter $\lambda\in\mathbb{N}$,
\[
\frac{1}{2}
\left\|
    \rho^{(0)}_{\mathrm{View}_M,T}
    -
    \rho^{(1)}_{\mathrm{View}_M,T}
\right\|_1
\leq \mathrm{negl}(\lambda).
\]
\end{definition}

\subsection{Good reconstruction boxes}
\label{def:good-box-formal}

\begin{definition}[Good reconstruction boxes (formal)]
Let $\lambda\in\mathbb{N}$, let $m,f,k\in\mathbb{N}$ satisfy
$0<f<k\leq m$, and let
$\mathcal{I}_{\mathrm{idx}}\subseteq[m]$ with
$|\mathcal{I}_{\mathrm{idx}}|=f$. Let $|\psi\rangle$ be a secret state and
let $\rho:=|\psi\rangle\langle\psi|$.

A reconstruction box $\mathcal{R}$ is
$(m,k,f,|\psi\rangle,\epsilon)$-good if, for uniformly random
$Q\subseteq[m]\setminus\mathcal{I}_{\mathrm{idx}}$ with $|Q|=k-f$,
\[
    \Pr\!\left[
        \rho' \leftarrow \mathcal{R}(\{|sh_i\rangle\}_{i\in Q})
        :
        \frac{1}{2}\|\rho'-\rho\|_1 \leq \mathrm{negl}(\lambda)
    \right]
    \geq \epsilon .
\]
The probability is taken over the random choice of $Q$, the randomness of $\mathcal{R}$, and the randomness used to generate the shares.
\end{definition}

\subsection{Traceability}
\label{def:trace1-formal}
We define the following tracing experiment, which formalizes the traceability of the quantum secret-sharing scheme. 

Let $\mathcal{A}$ be a QPT adversary, let $|\psi\rangle \in \mathcal{H}_\lambda$ be an arbitrary secret state, and let $\epsilon \in [0,1]$. Tags serve as unique identifiers of participants.

\experiment{$\mathrm{ExpTrace}_{\mathcal{A},\mathrm{LVT\text{-}QSS},\epsilon}(\lambda)$}
\begin{enumerate}
    \item $(m,k,\mathcal{I}_{\mathrm{idx}},st) \leftarrow \mathcal{A}(1^\lambda)$, where
    $\mathcal{I}_{\mathrm{idx}}=\{i_1,\ldots,i_f\}\subseteq [m]$ denotes the set of corrupted parties.

    \item $(tk,Tags) \leftarrow \mathrm{Setup}(1^\lambda,m)$, where
    $Tags=\{tag_i\}_{i=0}^{m-1}$. Define
    $\mathcal{I}:=\{tag_i : i\in \mathcal{I}_{\mathrm{idx}}\}$.
    The set $Tags$ is hidden from $\mathcal{A}$.

    \item $(\{|sh_i\rangle\}_{i=0}^{m-1},\{vk_i\}_{i=0}^{m-1})
    \leftarrow \mathrm{GenSecretShares}(|\psi\rangle,m,k,Tags,tk)$,
    and let $vk:=(vk_0,\ldots,vk_{m-1})$.

    \item $\mathcal{R} \leftarrow \mathcal{A}\bigl(st,\{|sh_i\rangle\}_{i\in\mathcal{I}_{\mathrm{idx}}}\bigr)$.

    \item If $\mathcal{R}$ is not $(m,k,f,|\psi\rangle,\epsilon)$-good, output $0$.

    \item Sample $Q\subseteq [m]\setminus\mathcal{I}_{\mathrm{idx}}$ uniformly at random subject to $|Q|=k-f$, and compute
    \[
    \mathcal{I}' \leftarrow \mathrm{SecretTrace}^{\mathcal{R}(\cdot)}\bigl(tk,\{|sh_i\rangle\}_{i\in Q}\bigr).
    \]

    \item Output $0$ if $\mathcal{I}'=\mathcal{I}$ and $\mathrm{SecretVerify}(vk,\mathcal{I}')=1$; otherwise output $1$.
\end{enumerate}
\noindent\rule{\linewidth}{0.5pt}

\begin{definition}[Traceability]
\label{def:formal-traceability}
For any QPT adversary \(\mathcal{A}\) and \(\epsilon\in[0,1]\), define its tracing advantage as
\begin{equation}
\mathrm{Adv}^{\mathrm{trace}}_{\mathcal{A},\mathrm{LVT\text{-}QSS},\epsilon}(\lambda)
:=
\epsilon\cdot
\Pr\!\left[
\mathrm{ExpTrace}_{\mathcal{A},\mathrm{LVT\text{-}QSS},\epsilon}(\lambda)=1
\right].
\end{equation}
The LVT-QSS scheme is formally traceable if the tracing advantage of every QPT adversary \(\mathcal{A}\) is negligible in \(\lambda\).
\end{definition}

\subsection{Non-imputability}
\label{def:ni-formal}
We define the following non-imputability experiment.
Let \(\mathcal{A}\) be a QPT adversary, and let \(|\psi\rangle\in\mathcal{H}_\lambda\) be
an arbitrary secret state. Tags serve as unique identifiers of participants.

\experiment{$\mathrm{ExpNI}_{\mathcal{A},\mathrm{LVT\text{-}QSS}}(\lambda)$}
\begin{enumerate}
    \item \((m,k,i^*,st)\leftarrow \mathcal{A}(1^\lambda)\), where
    \(i^*\in[m]\) denotes the honest participant that \(\mathcal{A}\) aims to frame.

    \item \((tk,Tags)\leftarrow \mathrm{Setup}(1^\lambda,m)\), where
    \(Tags=\{tag_i\}_{i=0}^{m-1}\). Let \(tag^*:=tag_{i^*}\).

    \item \((\{|sh_i\rangle\}_{i=0}^{m-1},\{vk_i\}_{i=0}^{m-1})
    \leftarrow \mathrm{GenSecretShares}(|\psi\rangle,m,k,Tags,tk)\),
    and let \(vk:=(vk_0,\ldots,vk_{m-1})\).

    \item \(\mathcal{I}^*_{\mathrm{tag}} \xleftarrow{\$}
    \mathcal{A}(st,\{|sh_i\rangle\}_{i\in[m]\setminus\{i^*\}},tk,vk)\).

    \item Output \(1\) if \(tag^*\in \mathcal{I}^*_{\mathrm{tag}}\) and
    \(\mathrm{SecretVerify}(vk,\mathcal{I}^*_{\mathrm{tag}})=1\); otherwise output \(0\).
\end{enumerate}
\noindent\rule{\linewidth}{0.5pt}

\begin{definition}[Non-imputability]
\label{def:formal-ni}
For any QPT adversary \(\mathcal{A}\), define its non-imputability advantage by
\begin{equation}
\mathrm{Adv}^{\mathrm{ni}}_{\mathcal{A},\mathrm{LVT\text{-}QSS}}(\lambda)
:=
\Pr\!\left[
\mathrm{ExpNI}_{\mathcal{A},\mathrm{LVT\text{-}QSS}}(\lambda)=1
\right].
\end{equation}
The LVT-QSS scheme is formally non-imputable if
\(\mathrm{Adv}^{\mathrm{ni}}_{\mathcal{A},\mathrm{LVT\text{-}QSS}}(\lambda)\)
is negligible in \(\lambda\) for every QPT adversary \(\mathcal{A}\).
\end{definition}

\section{Proofs of LVT-QSS}
\label{appendix:proof}

\subsection{Proof of Theorem~\ref{th:correct1}}
\label{app:lvt-correctness-proof}
\begin{proof}
    We show that any set of at least \(k\) valid shares enables correct reconstruction of the secret quantum state.

    \textbf{Linear structure of the coherent encoding.}
    In the \(\mathrm{GenSecretShares}\) procedure, the quantum secret
    \(|\psi\rangle = \sum_s \alpha_s |s\rangle\) is coherently encoded and distributed among the \(m\) parties' registers. Specifically, there exists an isometry \(V\) such that
    \begin{equation}
        V\left(\sum_s \alpha_s |s\rangle\right)
        =
        \sum_s \alpha_s |E(s)\rangle,
    \end{equation}
    where \(|E(s)\rangle\) denotes the coherent encoding of the basis element \(|s\rangle\).

    Let \(x_0,\ldots,x_{m-1}\in\mathbb{F}\) be distinct elements, and define the Vandermonde matrix
    \begin{equation}
        V_m =
        \begin{bmatrix}
        1 & x_0 & x_0^2 & \cdots & x_0^{k-1}\\
        \vdots & \vdots & \vdots & & \vdots\\
        1 & x_{m-1} & x_{m-1}^2 & \cdots & x_{m-1}^{k-1}
        \end{bmatrix}
        \in \mathbb{F}^{m\times k}.
    \end{equation}
    For any coefficient vector
    \(c=(c_0,\ldots,c_{k-1})^{\mathrm T}\in\mathbb{F}^k\), define the polynomial
    \begin{equation}
        p_c(t)=c_0+c_1t+\cdots+c_{k-1}t^{k-1}.
    \end{equation}

    The encoding in \(\mathrm{GenSecretShares}\) is expressed as
    \begin{equation}
        |\psi\rangle
        \longmapsto
        \sum_{s\in\mathbb{F}}\alpha_s
        \sum_{\substack{c\in\mathbb{F}^k\\ c_0=uk,\, c_{k-1}=s}}
        \bigotimes_{j=0}^{m-1}
        \bigl|p_c(x_j),r_j\bigr\rangle,
    \end{equation}
    where \(y_j=p_c(x_j)\), and \(uk\xleftarrow{\$}\mathbb{F}\). Since the points
    \(x_0,\ldots,x_{m-1}\) are distinct, the Vandermonde matrix \(V_m\) has full rank, and any set of \(k\) rows forms an invertible submatrix. Therefore, any \(k\) evaluations \(\{y_j\}\) uniquely determine the polynomial \(p_c(t)\) and the coefficient vector \(c\).

    \textbf{Reconstruction from any \(k\) shares.}
    Let \(P'\subseteq\{0,\ldots,m-1\}\) be any set of \(k\) parties. Without loss of generality, take \(P'=\{0,1,\ldots,k-1\}\). The basis state on the corresponding registers \(\{R_0,\ldots,R_{k-1}\}\) is
    \begin{equation}
    \begin{split}
        |y_0,\ldots,y_{k-1}\rangle
        &=
        |p_c(x_0),\ldots,p_c(x_{k-1})\rangle \\
        &=
        |V_k(x_0,\ldots,x_{k-1})\cdot c\rangle,
    \end{split}
    \end{equation}
    where \(V_k(x_0,\ldots,x_{k-1})\in\mathbb{F}^{k\times k}\) is the corresponding Vandermonde submatrix. Since \(V_k\) is invertible, we have
    \begin{equation}
        c
        =
        V_k^{-1}(x_0,\ldots,x_{k-1})
        \cdot
        (y_0,\ldots,y_{k-1})^{\mathrm T}.
    \end{equation}

    Applying the unitary implementation of this invertible linear transformation to the registers \(\{R_0,\ldots,R_{k-1}\}\) gives
    \begin{equation}
        \bigotimes_{j=0}^{k-1}|y_j\rangle
        \xmapsto{V_k^{-1}}
        |c_0,\ldots,c_{k-1}\rangle.
    \end{equation}
    In particular, the highest-degree coefficient satisfies \(c_{k-1}=s\).

    Next, the \(\mathrm{SecretRecover}\) procedure applies a cyclic right-shift permutation \(S_r\) to the registers \(\{R_0,\ldots,R_{k-1}\}\) and \(\{A_0,\ldots,A_{k-1}\}\):
    \begin{multline}
        |c_0,\ldots,c_{k-1};r_0,\ldots,r_{k-1}\rangle \\
        \xmapsto{S_r}
        |c_{k-1},c_0,\ldots,c_{k-2};
        r_{k-1},r_0,\ldots,r_{k-2}\rangle .
    \end{multline}
    By construction, the subsequent reversible linear transformations and controlled additions act only on \(\{R_1,\ldots,R_{k-1}\}\), together with the auxiliary registers, and do not affect \(R_0\). Consequently, register \(R_0\) contains
    \[
        \sum_s\alpha_s |s\rangle = |\psi\rangle,
    \]
    showing that \(\mathrm{SecretRecover}\) correctly reconstructs the secret from any \(k\) valid shares.

    Finally, let \(P'\subseteq [m]\) be any authorized set with \(|P'|\ge k\). The recovery procedure selects any \(k\)-subset \(P''\subseteq P'\) and applies the above reconstruction map to the shares \(\{|sh_i\rangle\}_{i\in P''}\). Let
    \[
        |\psi'\rangle
        \leftarrow
        \mathrm{SecretRecover}(\{|sh_i\rangle\}_{i\in P'})
    \]
    and \(\rho' := |\psi'\rangle\langle\psi'|\). Then
    \begin{equation}
        \Pr\left[
        \mathrm{Td}(\rho,\rho')\le \mathrm{negl}(\lambda)
        \right]
        \ge
        1-\mathrm{negl}(\lambda).
    \end{equation}
    Therefore, the \((k,m)\)-threshold LVT-QSS scheme is correct.
\end{proof}

\subsection{Proof of Theorem~\ref{th:secrecy}}
\label{app:lvt-secrecy-proof}

\begin{proof}
    We prove that any set of fewer than \(k\) parties learns no information about the secret quantum state.

    \textbf{Purification and encoding.}
    All internal randomness in the protocol is purified into an environment register \(E\). 
    The encoding procedure is realized by an isometry
    \begin{equation}
        U_{\mathrm{enc}} : \mathcal{H}_\lambda \to
        \left(\bigotimes_{i=0}^{m-1} \mathcal{H}_{R_i}\right)
        \otimes
        \left(\bigotimes_{i=0}^{m-1} \mathcal{H}_{A_i}\right)
        \otimes \mathcal{H}_E .
    \end{equation}

    Let \(q=|\mathbb{F}|\), and let the secret state be
    \[
        |\psi\rangle=\sum_s \alpha_s |s\rangle,
        \qquad
        \sum_s |\alpha_s|^2=1 .
    \]
    In the encoding, the secret basis value \(s\) is embedded as the highest-degree coefficient \(c_{k-1}=s\). 
    The remaining coefficients $c_0,c_1,\ldots,c_{k-2}$ are sampled uniformly from \(\mathbb{F}\) and independently of \(s\), with \(c_0=uk\). 
    For a coefficient vector \(c=(c_0,\ldots,c_{k-1})\), define
    \begin{equation}
        p_c(t)=\sum_{\ell=0}^{k-1} c_\ell t^\ell .
    \end{equation}
    The encoded global state can be written as
    \begin{multline}
    |\Psi\rangle
    =
    \frac{1}{\sqrt{q^{k-1}}}
    \sum_s \alpha_s 
    \sum_{\substack{c\in\mathbb{F}^k\\ c_{k-1}=s}} \\
    \Big(\bigotimes_{i=0}^{m-1} |p_c(x_i)\rangle_{R_i}\Big)
    \otimes
    \Big(\bigotimes_{i=0}^{m-1} |r_i\rangle_{A_i}\Big)
    \otimes |\xi_c\rangle_E .
    \end{multline}
    Since the classical randomness \(c_0,\ldots,c_{k-2}\) can be purified into the environment, we assume without loss of generality that the states
    \(\{|\xi_c\rangle_E\}_c\) are pairwise orthogonal.

    \textbf{Reduced state on an unauthorized set.}
    Let \(M\subseteq [m]\) be any unauthorized set with \(|M|<k\). 
    Tracing out the environment register eliminates the cross terms, since
    \[
        \mathrm{Tr}_E\!\left(|\xi_c\rangle\langle \xi_{c'}|\right)
        =
        \delta_{c,c'} .
    \]
    Further tracing out all registers outside \(M\) gives the reduced state on the adversary's view:
    \begin{multline}
    \rho_{\mathrm{View}_M}(|\psi\rangle\langle\psi|)
    =
    \frac{1}{q^{k-1}}
    \sum_s |\alpha_s|^2
    \sum_{\substack{c\in\mathbb{F}^k\\ c_{k-1}=s}}\\
    \Big(
    \bigotimes_{i\in M}
    |p_c(x_i)\rangle\langle p_c(x_i)|
    \Big)
    \otimes \rho_{A_M},
    \end{multline}
    where
    \[
        \rho_{A_M}:=\bigotimes_{i\in M}|r_i\rangle\langle r_i|
    \]
    is independent of the secret.

    Define the linear map
    \[
        T_M:\mathbb{F}^k\to\mathbb{F}^{|M|},
        \qquad
        T_M(c)=(p_c(x_i))_{i\in M}.
    \]
    Fix \(s\in\mathbb{F}\), and impose only the constraint \(c_{k-1}=s\). 
    The remaining variables $c_0,c_1,\ldots,c_{k-2}$ are uniformly random over \(\mathbb{F}^{k-1}\) and independent of \(s\). 
    Since the points \(\{x_i\}_{i\in M}\) are distinct and \(|M|<k\), the corresponding Vandermonde submatrix has full row rank \(|M|\). 
    Therefore, for every \(y\in\mathbb{F}^{|M|}\), the system
    \[
        T_M(c)=y,
        \qquad
        c_{k-1}=s
    \]
    has exactly \(q^{k-1-|M|}\) solutions in the variables \(c_0,\ldots,c_{k-2}\), independent of \(s\) and \(y\). 
    Hence
    \begin{multline}
        \sum_{\substack{c\in\mathbb{F}^k\\ c_{k-1}=s}}
        \bigotimes_{i\in M}
        |p_c(x_i)\rangle\langle p_c(x_i)|
        \\
        =
        q^{k-1-|M|}
        \sum_{y\in\mathbb{F}^{|M|}}
        |y\rangle\langle y| .
    \end{multline}

    Substituting this expression into the reduced state and using
    \(\sum_s|\alpha_s|^2=1\), we obtain
    \begin{multline}
        \rho_{\mathrm{View}_M}(|\psi\rangle\langle\psi|)
        =
        \frac{1}{q^{k-1}}
        \sum_s |\alpha_s|^2
        q^{k-1-|M|}
        \sum_{y\in\mathbb{F}^{|M|}}
        |y\rangle\langle y|
        \otimes \rho_{A_M}
        \\
        =
        \left(
        \frac{1}{q^{|M|}}
        \sum_{y\in\mathbb{F}^{|M|}}
        |y\rangle\langle y|
        \right)
        \otimes \rho_{A_M}.
    \end{multline}
    Denote this fixed state by
    \begin{equation}
        \sigma_M
        :=
        \left(
        \frac{1}{q^{|M|}}
        \sum_{y\in\mathbb{F}^{|M|}}
        |y\rangle\langle y|
        \right)
        \otimes \rho_{A_M}.
    \end{equation}
    Crucially, \(\sigma_M\) is independent of the amplitudes \(\{\alpha_s\}\) and therefore independent of the secret state.

    Thus, for any two secret states \(|\psi_0\rangle\) and \(|\psi_1\rangle\), the corresponding adversarial views satisfy
    \begin{equation}
        \rho^{(0)}_{\mathrm{View}_M}
        =
        \rho^{(1)}_{\mathrm{View}_M}
        =
        \sigma_M .
    \end{equation}
    Consequently,
    \begin{equation}
    \frac{1}{2}\left\|
    \rho^{(0)}_{\mathrm{View}_M}
    -
    \rho^{(1)}_{\mathrm{View}_M}
    \right\|_1
    =
    0
    \le \mathrm{negl}(\lambda).
    \end{equation}
    Thus every unauthorized subset \(M\) with \(|M|<k\) has negligible distinguishing advantage between any two secret states. Therefore, the LVT-QSS scheme satisfies secrecy.
\end{proof}

\subsection{Proof of Theorem~\ref{th:tsecrecy}}
\label{app:lvt-tracer-secrecy-proof}

\begin{proof}
Let \(M\subseteq[m]\) be an unauthorized subset with \(|M|<k\), and let \(T\)
denote the classical register storing the tracing key \(tk\) generated by
\(\mathrm{Setup}(1^\lambda,m)\).

By Theorem~\ref{th:secrecy}, the unauthorized view
\(\mathrm{View}_M=R_M\otimes A_M\), without the tracing-key register \(T\), is
independent of the secret state. Hence, for any two secret states
\(|\psi_0\rangle\) and \(|\psi_1\rangle\), there exists a density operator
\(\sigma_M\), depending only on \(M\) and the public parameters, such that
\begin{equation}
\rho^{(0)}_{\mathrm{View}_M}
=
\rho^{(1)}_{\mathrm{View}_M}
=
\sigma_M .
\end{equation}

It remains to check that adding the tracing-key register \(T\) does not
introduce any dependence on the secret. The tracing key \(tk\) is generated by
\(\mathrm{Setup}(1^\lambda,m)\) independently of the secret input. Moreover, the
auxiliary registers \(A_M\) are generated from \(tk\), local randomness, and the
user key, all of which are independent of which secret state is shared.
Therefore, conditioned on any fixed tracing key \(tk\), the unauthorized view is
still independent of the secret. That is, for every \(tk\),
\begin{equation}
\rho^{(0)}_{\mathrm{View}_M\mid tk}
=
\rho^{(1)}_{\mathrm{View}_M\mid tk}.
\end{equation}

Since \(T\) is a classical register storing \(tk\), the joint state on
\(\mathrm{View}_M\otimes T\) can be written as
\begin{equation}
\rho^{(b)}_{\mathrm{View}_M,T}
=
\sum_{tk} p(tk)
|tk\rangle\!\langle tk|_T
\otimes
\rho^{(b)}_{\mathrm{View}_M\mid tk}, b\in\{0,1\}.
\end{equation}
Using the equality of the conditional views for every \(tk\), we obtain
\begin{equation}
\rho^{(0)}_{\mathrm{View}_M,T}
=
\rho^{(1)}_{\mathrm{View}_M,T}.
\end{equation}
Consequently,
\begin{equation}
\frac{1}{2}
\left\|
\rho^{(0)}_{\mathrm{View}_M,T}
-
\rho^{(1)}_{\mathrm{View}_M,T}
\right\|_1
=
0
\le
\mathrm{negl}(\lambda).
\end{equation}
Thus every unauthorized subset \(M\) with \(|M|<k\) has negligible distinguishing advantage between any two secret states even when the
tracing key is included in its view. Therefore, the \((k,m)\)-threshold LVT-QSS
scheme satisfies tracer secrecy.
\end{proof}

\subsection{Proof of Theorem~\ref{th:lvt-trace}}
\label{app:lvt-trace-proof}

\noindent\textbf{Theorem~\ref{th:lvt-trace} }(Formal).
{\itshape
Assume that $H:\mathbb{F}\to\mathbb{F}$ is quantum-secure collision resistant and preimage resistant, and that the symmetric-key encryption scheme
$\Pi_{\mathrm{SKE}}=(\mathrm{Enc},\mathrm{Dec})$ satisfies the sparse-hit property with parameter $\alpha_{\max}(\lambda)$.
Then, for any QPT adversary $\mathcal{A}$, there exist a QPT algorithm $\mathcal{B}_1$ and a negligible function $\mathrm{negl}(\cdot)$ such that, for security parameter $\lambda\in\mathbb{N}$,
\begin{multline}
\mathrm{Adv}^{\mathrm{trace}}_{\mathcal{A},\mathrm{LVT\text{-}QSS},\epsilon}(\lambda) \\
\leq 
\epsilon(\lambda)
\left(
\mathrm{Adv}^{\mathrm{coll}}_{\mathcal{B}_1,H}(\lambda)
+
\ell(m-f)\alpha_{\max}(\lambda)
+
\mathrm{negl}(\lambda)
\right),
\end{multline}
where $\ell := |I\setminus I''|$. Consequently, if $\ell(m-f)=\mathrm{poly}(\lambda)$ and $\alpha_{\max}(\lambda)$ is negligible, then the LVT-QSS scheme is traceable.
}

\begin{proof}
Let $\mathcal{I}_{\mathrm{idx}}\subseteq [m]$ be the set of leaked indices, with $|\mathcal{I}_{\mathrm{idx}}|=f$, and let $\mathcal{I}:=\{\,tag_i \mid i\in\mathcal{I}_{\mathrm{idx}}\,\}\subseteq Tags$
be the corresponding leaked tag set. Let $\mathcal{I}'$ denote the output of $\mathrm{SecretTrace}$, and let $ok:=\mathrm{SecretVerify}(vk,\mathcal{I}')\in\{0,1\}.$

The traceability experiment fails precisely when the protocol executes correctly, but either verification rejects or it accepts an incorrect leaked tag set. Accordingly, define
\begin{align}
E_1 &:= \{\,G\wedge ok=0\,\},\\
E_2 &:= \{\,G\wedge ok=1\wedge \mathcal{I}'\neq \mathcal{I}\,\}.
\end{align}
Then
\begin{equation}
\Pr[\mathrm{ExpTrace}=1]\le \Pr[E_1]+\Pr[E_2].
\end{equation}

By completeness of $\mathrm{SecretVerify}$, one has
\begin{equation}
\Pr[E_1]\le \mathrm{negl}(\lambda).
\end{equation}

Next decompose $E_2$ into
\begin{align}
E_2^{\mathrm{forge}}
&:=
\{\,G\wedge ok=1\wedge \exists\,tag^*\in\mathcal{I}' : tag^*\notin Tags\,\},\\
E_2^{\mathrm{wrong}}
&:=
\{\,G\wedge ok=1\wedge \mathcal{I}'\subseteq Tags \wedge \mathcal{I}'\neq \mathcal{I}\,\}.
\end{align}
Thus
\begin{equation}
\Pr[E_2]\le \Pr[E_2^{\mathrm{forge}}]+\Pr[E_2^{\mathrm{wrong}}].
\end{equation}

If $E_2^{\mathrm{forge}}$ occurs, then some $tag^*\in\mathcal{I}'\setminus Tags$ is accepted by $\mathrm{SecretVerify}$. By the definition of $\mathrm{SecretVerify}$, this implies $H(tag^*)\in \{vk_i\}_{i=0}^{m-1},$
and hence $H(tag^*)=H(tag_j)$ for some $tag_j\in Tags$. Since $tag^*\neq tag_j$, this yields a collision in $H$. Therefore,
\begin{equation}
\Pr[E_2^{\mathrm{forge}}]\le \mathrm{Adv}^{\mathrm{coll}}_{\mathcal{B}_1,H}(\lambda).
\end{equation}

Now suppose that $E_2^{\mathrm{wrong}}$ occurs. Since $\mathcal{I}'\subseteq Tags$ but $\mathcal{I}'\neq \mathcal{I}$ and $|\mathcal{I}'|=|\mathcal{I}|=f$, there exists some $i\notin\mathcal{I}_{\mathrm{idx}}$ such that $tag_i\in\mathcal{I}'$. By the tracing rule, this implies that for some $r\in I\setminus I''$, the first-layer decryption output $\mathrm{Dec}(r,uk)$ falls into the set $S_i$ of ciphertexts that decrypt to $tag_i$ under the tracing key. By the sparse-hit property,
\begin{equation}
\Pr\bigl[\mathrm{Dec}(r,uk)\in S_i\bigr]\le |S_i|\,\alpha_{\max}(\lambda).
\end{equation}
Let $l:=|I\setminus I''|$. A union bound over all $r\in I\setminus I''$ and all $i\notin\mathcal{I}_{\mathrm{idx}}$ gives
\begin{equation}
\Pr[E_2^{\mathrm{wrong}}]\le l(m-f)\,\alpha_{\max}(\lambda),
\end{equation}
absorbing the polynomial bound on $|S_i|$ into $\alpha_{\max}(\lambda)$.

Combining the above estimates, we obtain
\begin{multline}
\Pr\!\left[
\mathrm{ExpTrace}_{\mathcal{A},\mathrm{LVT\text{-}QSS},\epsilon}(\lambda)=1
\right]
\le
\mathrm{Adv}^{\mathrm{coll}}_{\mathcal{B}_1,H}(\lambda)
\\
+
l(m-f)\,\alpha_{\max}(\lambda)
+
\mathrm{negl}(\lambda).
\end{multline}
Therefore,
\begin{multline}
\mathrm{Adv}^{\mathrm{trace}}_{\mathcal{A},\mathrm{LVT\text{-}QSS},\epsilon}(\lambda)
=
\epsilon
\Pr\!\left[
\mathrm{ExpTrace}_{\mathcal{A},\mathrm{LVT\text{-}QSS},\epsilon}(\lambda)=1
\right] \\
\le
\epsilon\Bigl(
\mathrm{Adv}^{\mathrm{coll}}_{\mathcal{B}_1,H}(\lambda)
+l(m-f)\,\alpha_{\max}(\lambda)
+\mathrm{negl}(\lambda)\Bigr).
\end{multline}

Since $\mathrm{Adv}^{\mathrm{coll}}_{\mathcal{B}_1,H}(\lambda)$ is negligible, $l(m-f)=\mathrm{poly}(\lambda)$, and $\alpha_{\max}(\lambda)$ is negligible, the right-hand side is negligible. Therefore, the LVT-QSS scheme is traceable.
\end{proof}

\subsection{Proof of Theorem~\ref{th:lvt-ni}}
\label{app:lvt-ni-proof}

\noindent\textbf{Theorem~\ref{th:lvt-ni} }(Formal).
{\itshape
Assume that $H:\mathbb{F}\to\mathbb{F}$ is quantum-secure collision-resistant and preimage-resistant. Then, for any QPT adversary $\mathcal{A}$, there exist QPT algorithms $\mathcal{B}_2^{\mathrm{pre}}$ and $\mathcal{B}_2^{\mathrm{coll}}$, and a negligible function $\mathrm{negl}(\cdot)$, such that, for security parameter $\lambda\in\mathbb{N}$,
\begin{equation}
\mathrm{Adv}^{\mathrm{ni}}_{\mathcal{A},\mathrm{LVT\text{-}QSS}}(\lambda)
\leq
\mathrm{Adv}^{\mathrm{pre}}_{\mathcal{B}_2^{\mathrm{pre}},H}(\lambda)
+
\mathrm{Adv}^{\mathrm{coll}}_{\mathcal{B}_2^{\mathrm{coll}},H}(\lambda)
+
\mathrm{negl}(\lambda).
\end{equation}
Consequently, the LVT-QSS scheme is non-imputable.
}

\begin{proof}
Let \(\mathcal{A}\) be an arbitrary QPT adversary in the non-imputability experiment \(\mathrm{ExpNI}_{\mathcal{A},\mathrm{LVT\text{-}QSS}}(\lambda)\). We show that any successful adversary yields either a preimage or a collision for \(H\).

\textbf{Game 0.}
This is the real experiment \(\mathrm{ExpNI}_{\mathcal{A},\mathrm{LVT\text{-}QSS}}(\lambda)\). Setup generates pairwise distinct tags \(Tags=\{tag_i\}_{i\in[m]}\), and the verification keys are \(vk_i:=H(tag_i)\).

\textbf{Game 1.}
This game is identical to Game~0 except that, after
\(\mathrm{GenSecretShares}\), the verification key of the challenged honest party \(i^*\) is replaced by \(vk_{i^*}:=y^*\), where \(y^*=H(w^*)\) and \(w^*\xleftarrow{\$}\mathbb{F}\) is sampled independently.

Because \(tag_{i^*}\) and \(w^*\) are sampled from the same distribution, the values \(H(tag_{i^*})\) and \(H(w^*)\) are identically distributed. Therefore, the adversary's views in Games~0 and~1 are identically distributed, and
\begin{equation}
\Pr[\mathrm{ExpNI}_{\mathcal{A},\mathrm{LVT\text{-}QSS}}(\lambda)=1]
=
\Pr[\mathrm{ExpNI}^{(1)}_{\mathcal{A},\mathrm{LVT\text{-}QSS}}(\lambda)=1].
\end{equation}

Suppose that \(\mathrm{ExpNI}^{(1)}_{\mathcal{A},\mathrm{LVT\text{-}QSS}}(\lambda)=1\).
Then the adversary outputs a set \(I^*\) containing a tag \(tag^*\) attributed to the challenged honest party \(i^*\), and \(\mathrm{SecretVerify}(vk,I^*)=1\). By the definition of \(\mathrm{SecretVerify}\),
\[
H(tag^*)\in \{vk_i\}_{i=0}^{m-1}.
\]
In Game~1,
\[
\{vk_i\}_{i=0}^{m-1}
=
\{y^*\}\cup \{H(tag_i)\}_{i\neq i^*}.
\]

There are two cases. If \(H(tag^*)=y^*\), then \(tag^*\) is a valid preimage of the random image \(y^*\). This event can therefore be used to construct a QPT algorithm \(\mathcal{B}_2^{\mathrm{pre}}\) that, on input \(y^*\), outputs \(tag^*\) as a valid preimage, thereby violating the preimage resistance of \(H\). Hence,
\begin{equation}
\Pr[H(tag^*)=y^*]
\le
\mathrm{Adv}^{\mathrm{pre}}_{\mathcal{B}_2^{\mathrm{pre}},H}(\lambda).
\end{equation}
If \(H(tag^*)=H(tag_i)\) for some \(i\neq i^*\), then, since the tags are distinct, \(tag^*\neq tag_i\), and the pair \((tag^*,tag_i)\) is a collision for \(H\). This event can therefore be used to construct a QPT algorithm \(\mathcal{B}_2^{\mathrm{coll}}\) that outputs this pair, thereby violating the collision resistance of \(H\). Thus,
\begin{equation}
\Pr[\exists i\neq i^*: H(tag^*)=H(tag_i)]
\le
\mathrm{Adv}^{\mathrm{coll}}_{\mathcal{B}_2^{\mathrm{coll}},H}(\lambda).
\end{equation}

Since a successful execution of Game~1 implies one of these two cases, we obtain
\begin{multline}
\mathrm{Adv}^{\mathrm{ni}}_{\mathcal{A},\mathrm{LVT\text{-}QSS}}(\lambda)
\\
\le
\mathrm{Adv}^{\mathrm{pre}}_{\mathcal{B}_2^{\mathrm{pre}},H}(\lambda)
+
\mathrm{Adv}^{\mathrm{coll}}_{\mathcal{B}_2^{\mathrm{coll}},H}(\lambda).
\end{multline}
Because \(H\) is quantum-secure preimage resistant and collision resistant, the right-hand side is negligible in \(\lambda\). Therefore, the LVT-QSS scheme is non-imputable.
\end{proof}

\section{Formal Definitions of MTA-QSS}
\label{appendix:def-mta-qss}

\subsection{Correctness}
\label{def:correct2-formal}

\begin{definition}[Correctness (formal)]
An MTA-QSS scheme is correct if there exists a negligible function
$\mathrm{negl}(\cdot)$ such that, for security parameter
$\lambda\in\mathbb{N}$, number of parties $m\in\mathbb{N}$,
threshold $k$ with $1\le k\le m$, subset $P\subseteq [m]$ with
$|P|\ge k$, and secret state $|\psi\rangle\in\mathcal{H}_\lambda$,
the following holds.

Let $\rho:=|\psi\rangle\langle\psi|$. Generate the shares and run the
reconstruction algorithm as
\begin{align}
(tk,Tags)
&\leftarrow \mathrm{Setup}(1^\lambda,m), \nonumber\\
(sh,vk,uk,ck)
&\leftarrow \mathrm{GenSecretShares}(|\psi\rangle,m,k,Tags,tk),
\nonumber\\
out
&\leftarrow \mathrm{SecretRecover}(\{|sh_i\rangle\}_{i\in P}).
\end{align}
where $tk := \{tk_i\}_{i=0}^{m-1}$, $sh := \{|sh_i\rangle\}_{i=0}^{m-1}$, and $vk := \{vk_i\}_{i=0}^{m-1}$.

Then
\begin{equation}
\Pr\left[
    out=(\rho',\cdot)
    \ \text{and}\
    \frac{1}{2}\|\rho'-\rho\|_1\leq \mathrm{negl}(\lambda)
\right]
\geq 1-\mathrm{negl}(\lambda),
\end{equation}
where the probability is taken over the internal randomness of
$\mathrm{Setup}$, $\mathrm{GenSecretShares}$, and $\mathrm{SecretRecover}$.

This requirement applies only when all shares used in reconstruction are honestly generated by $\mathrm{GenSecretShares}$. No correctness guarantee is required when some of the input shares are forged or adversarially corrupted.
\end{definition}

\subsection{Tracer secrecy}
\label{def:tsecrecy2-formal}

\begin{definition}[Tracer secrecy (formal)]
Let $\mathcal{H}=\{\mathcal{H}_\lambda\}_{\lambda\in\mathbb{N}}$ denote the secret-state spaces of an MTA-QSS scheme.  
For every $M\subseteq[m]$ with $|M|<k$ and any two secret states $|\psi_0\rangle,|\psi_1\rangle\in\mathcal{H}_\lambda$, define
\[
\mathrm{View}_M := \bigotimes_{i\in M} (R_i \otimes A_i),\quad
T := \bigotimes_{i=0}^{m-1} tk_i,\quad
Y := uk
\]
and let $\rho^{(b)}_{\mathrm{View}_M,T,Y}$ be the reduced state generated by sharing $|\psi_b\rangle$, for $b\in\{0,1\}$. Then there exists a negligible function $\mathrm{negl}(\cdot)$ such that, for security parameter $\lambda \in \mathbb{N}$,
\[
\frac{1}{2}
\left\|
    \rho^{(0)}_{\mathrm{View}_M,T,Y} - \rho^{(1)}_{\mathrm{View}_M,T,Y}
\right\|_1 \le \mathrm{negl}(\lambda).
\]
\end{definition}

\subsection{Traceability}
\label{def:trace2-formal}
We define the following tracing experiment.
Let $\mathcal{A}$ be a QPT adversary, let $|\psi\rangle\in\mathcal{H}_\lambda$ be an arbitrary secret state, and let $Tags=\{tag_i\}_{i=0}^{m-1}$, where $tag_i$ uniquely identifies party $i$.

\experiment{$\mathrm{ExpTrace}_{\mathcal{A},\mathrm{MTA\text{-}QSS}}(\lambda)$}
\begin{enumerate}
    \item $(m,k,V_{\mathrm{idx}},st) \leftarrow \mathcal{A}(1^\lambda)$, where $V_{\mathrm{idx}}=\{i_0,\dots,i_{t-1}\}\subseteq [m]$ denotes the set of participants corrupted by $\mathcal{A}$.

    \item $(\{tk_i\}_{i=0}^{m-1},Tags)\leftarrow \mathrm{Setup}(1^\lambda,m).$

    \item $(\{|sh_i\rangle\}_{i=0}^{m-1},\{vk_i\}_{i=0}^{m-1},uk,ck)\leftarrow
    \mathrm{GenSecretShares}(|\psi\rangle,m,k,Tags,\{tk_i\}_{i=0}^{m-1})$.

    \item Set $tk:=(tk_0,\dots,tk_{m-1})$, and $vk:=(vk_0,\dots,vk_{m-1}).$

    \item Sample a subset $U_{\mathrm{idx}}\subseteq [m]\setminus V_{\mathrm{idx}}$ uniformly at random subject to $|U_{\mathrm{idx}}|=k-t$. Define $P:=U_{\mathrm{idx}}\cup V_{\mathrm{idx}}$, and let
    $U:=\{tag_i : i\in U_{\mathrm{idx}}\}$ and
    $V:=\{tag_i : i\in V_{\mathrm{idx}}\}$.

    \item $\{\,|sh_i^{\mathrm{mal}}\rangle\,\}_{i\in V_{\mathrm{idx}}}
    \leftarrow
    \mathcal{A}\bigl(st,\{\,|sh_i\rangle\,\}_{i\in V_{\mathrm{idx}}}\bigr)$.
    For each $i\in P$, define
    \[
    |\widetilde{sh}_i\rangle :=
    \begin{cases}
    |sh_i^{\mathrm{mal}}\rangle, & i\in V_{\mathrm{idx}},\\
    |sh_i\rangle, & i\in U_{\mathrm{idx}}.
    \end{cases}
    \]

    \item $(|\psi'\rangle,I)\leftarrow \mathrm{SecretRecover}(\{\,|\widetilde{sh}_i\rangle\,\}_{i\in P})$.

    \item $(U',V')\leftarrow \mathrm{SecretTrace}(tk,uk,I)$.

    \item Output $0$ if $U'=U$ and $\mathrm{SecretVerify}(vk,uk,ck,U',V',I)=1$; otherwise output $1$.
\end{enumerate}
\noindent\rule{\linewidth}{0.5pt}

\begin{definition}[Traceability (formal)]
For any QPT adversary $\mathcal{A}$, define its tracing advantage as
\begin{equation}
    \mathrm{Adv}^{\mathrm{trace}}_{\mathcal{A},\mathrm{MTA\text{-}QSS}}(\lambda)
    := \Pr\!\bigl[\mathrm{ExpTrace}_{\mathcal{A},\mathrm{MTA\text{-}QSS}}(\lambda)=1
    \bigr].
\end{equation}
The MTA-QSS scheme is traceable if the tracing advantage of every QPT adversary is negligible in $\lambda$.
\end{definition}

\subsection{Non-imputability}
\label{def:ni2-formal}

We define the following non-imputability experiment.
Let $\mathcal{A}$ be a QPT adversary, let $\lambda$ be the security 
parameter, and let $|\psi\rangle\in\mathcal{H}_\lambda$ be an arbitrary secret quantum state.

\experiment{$\mathrm{ExpNI}_{\mathcal{A},\mathrm{MTA\text{-}QSS}}(\lambda)$}
\begin{enumerate}
    \item $(m, k, i^*, st) \leftarrow \mathcal{A}(1^\lambda)$, where $i^* \in [m]$ is the index of the honest party that $\mathcal{A}$ attempts to frame.

    \item $(\{tk_i\}_{i=0}^{m-1},Tags) \leftarrow \mathrm{Setup}(1^\lambda, m)$, where $Tags=(tag_0,\dots,tag_{m-1})$. Define $tag^* := tag_{i^*}$.

    \item $(\{|sh_i\rangle\}_{i=0}^{m-1}, \{vk_i\}_{i=0}^{m-1}, uk, ck) \leftarrow \mathrm{GenSecretShares}(|\psi\rangle, m, k, Tags, \{tk_i\}_{i=0}^{m-1})$.

    \item Set $tk := (tk_0,\dots,tk_{m-1})$ and $vk := (vk_0,\dots,vk_{m-1})$.

    \item Sample a subset $P_{\mathrm{rest}} \subseteq [m]\setminus\{i^*\}$ uniformly at random subject to $|P_{\mathrm{rest}}|=k-1$, and define
    \[
    P := P_{\mathrm{rest}} \cup \{i^*\}.
    \]
    Then compute
    \[
    (\cdot, I) \leftarrow \mathrm{SecretRecover}(\{|sh_i\rangle\}_{i\in P}).
    \]

    \item $(U^*, V^*) \leftarrow \mathcal{A}(st, \{|sh_i\rangle\}_{i\in[m]\setminus\{i^*\}}, tk, vk, I)$.

    \item Output $1$ if $tag^*\in V^*$ and
    $\mathrm{SecretVerify}(vk,uk,ck,U^*,V^*,I)=1$; otherwise output $0$.
\end{enumerate}
\noindent\rule{\linewidth}{0.5pt}

\begin{definition}[Non-imputability (formal)]
For any QPT adversary $\mathcal{A}$, define its non-imputability advantage as
\begin{equation}
    \mathrm{Adv}^{\mathrm{ni}}_{\mathcal{A},\mathrm{MTA\text{-}QSS}}(\lambda)
    := \Pr\!\bigl[
        \mathrm{ExpNI}_{\mathcal{A},\mathrm{MTA\text{-}QSS}}(\lambda)=1
    \bigr].
\end{equation}
The MTA-QSS scheme is non-imputable if the non-imputability advantage of every QPT adversary is negligible in $\lambda$.
\end{definition}

\section{Proofs of MTA-QSS}
\label{appendix:proof2}

\subsection{Proof of Theorem~\ref{th:correct2}}
\label{app:mta-correctness-proof}
\begin{proof}
It is sufficient to consider honestly generated shares. For such shares, the MTA-QSS scheme uses exactly the same coherent $(k,m)$-threshold encoding on the quantum share registers as the LVT-QSS scheme, and the reconstruction procedure $\mathrm{SecretRecover}$ applies the same decoding map on these registers.

The two schemes differ only in the auxiliary classical registers. In MTA-QSS, the auxiliary register associated with the $j$th share stores $(r_j\|tk_j)$, whereas in LVT-QSS it stores only $r_j$. These auxiliary registers do not affect the decoding operation on the quantum share registers and therefore do not alter the recovered secret state.

It follows that, for honestly generated shares, the reduced state on the output secret register in MTA-QSS is identical to that in LVT-QSS. Since correctness for LVT-QSS was established in Theorem~\ref{th:correct1}, the MTA-QSS scheme is correct.
\end{proof}

\subsection{Proof of Theorem~\ref{th:secrecy2}}
\label{app:mta-secrecy-proof}
\begin{proof}
Let $M\subset[m]$ be any unauthorized set with $|M|<k$. For each $b\in\{0,1\}$, let $\rho_{\mathrm{View}_M}^{(b)}$ denote the reduced state on the registers $R_M\otimes A_M$ when the shared secret is the input state corresponding to $b$.

On the quantum share registers, the MTA-QSS scheme uses exactly the same encoding as the LVT-QSS scheme. Hence, by the same argument as in the proof of Theorem~\ref{th:secrecy}, the reduced state on $R_M$ is independent of the secret for every unauthorized set $M$.

It remains to consider the auxiliary registers $A_M$. In MTA-QSS, these registers contain only classical values derived from protocol randomness, suchas tags and tracing-related data, and are independent of the input secret state. Consequently, the joint reduced state on $R_M\otimes A_M$ is also independent of the secret. In particular,
\begin{equation}
    \rho_{\mathrm{View}_M}^{(0)}
    =
    \rho_{\mathrm{View}_M}^{(1)} .
\end{equation}
Therefore,
\begin{equation}
    \frac{1}{2}
    \left\|
        \rho_{\mathrm{View}_M}^{(0)}
        -
        \rho_{\mathrm{View}_M}^{(1)}
    \right\|_{\mathrm{tr}}
    =
    0
    \leq \mathrm{negl}(\lambda).
\end{equation}
Hence, the $(k,m)$-threshold MTA-QSS scheme satisfies secrecy.
\end{proof}

\subsection{Proof of Theorem~\ref{th:tracer-secrecy2}}
\label{app:mta-tracer-secrecy-proof}

\begin{proof}
Fix any unauthorized set $M\subseteq[m]$ with $|M|<k$ and any two secret states
$|\psi_0\rangle,|\psi_1\rangle\in\mathcal H_\lambda$. Let
\[
    (\{tk_i\}_{i=0}^{m-1},\mathrm{Tags})
    \leftarrow \mathrm{Setup}(1^\lambda,m),
\]
set $tk:=(tk_0,\ldots,tk_{m-1})$, and let $T$ be the classical register storing
$tk$.

For each $b\in\{0,1\}$, run $(\{|sh_i^{(b)}\rangle\}_{i=0}^{m-1},\cdot,uk^{(b)},\cdot)
    \leftarrow
    \mathrm{GenSecretShares}(|\psi_b\rangle,m,k,\mathrm{Tags},tk),$
and let $Y$ be the classical register storing $uk^{(b)}$. Write $\rho^{(b)}_{\mathrm{View}_M,T,Y}$ for the reduced state on
\[
    \mathrm{View}_M\otimes T\otimes Y
    =
    R_M\otimes A_M\otimes T\otimes Y .
\]

By secrecy, the reduced state on the quantum share registers held by $M$ is independent of the input secret. The remaining registers in the above view are classical: the auxiliary registers contain values of the form $r_i\|tk_i$, the register $T$ contains the tracing keys, and $Y$ contains $uk^{(b)}$. These values are generated from protocol randomness and tracing data, independently of the input secret state. Hence, for every fixed value of $tk$, the conditional state on $\mathrm{View}_M\otimes Y$ is independent of $b$; denote it by $\rho_{\mathrm{View}_M,Y\mid tk}$.

Since $T$ is classical, the joint state has the cq decomposition
\begin{equation}
    \rho^{(b)}_{\mathrm{View}_M,T,Y}
    =
    \sum_{tk} p(tk)\,
    |tk\rangle\!\langle tk|_T
    \otimes
    \rho_{\mathrm{View}_M,Y\mid tk},
\end{equation}
which is independent of $b$. Therefore,
\begin{equation}
    \rho^{(0)}_{\mathrm{View}_M,T,Y}
    =
    \rho^{(1)}_{\mathrm{View}_M,T,Y},
\end{equation}
and hence
\begin{equation}
    \frac12
    \left\|
        \rho^{(0)}_{\mathrm{View}_M,T,Y}
        -
        \rho^{(1)}_{\mathrm{View}_M,T,Y}
    \right\|_1
    =
    0
    \leq \mathrm{negl}(\lambda).
\end{equation}
Thus, the MTA-QSS scheme satisfies tracer secrecy.
\end{proof}

\subsection{Proof of Theorem~\ref{th:trace2}}
\label{app:mta-tracer2-proof}

\noindent\textbf{Theorem~\ref{th:trace2} }(Formal).
{\itshape
The MTA-QSS scheme is traceable against any QPT adversary. For every QPT adversary $\mathcal{A}$, there exists a negligible function $\mathrm{negl}(\cdot)$ such that, for security parameter $\lambda \in \mathbb{N}$,
\[
    \mathrm{Adv}^{\mathrm{trace}}_{\mathcal{A},\mathrm{MTA\text{-}QSS}}(\lambda)
    \leq \mathrm{negl}(\lambda).
\]
}

\begin{proof}
Let $U_{\mathrm{idx}}$ be the index set of honest parties involved in reconstruction, and let $U\subseteq Tags$ denote the corresponding honest tag set. Let $(U',V')\leftarrow \mathrm{SecretTrace}(tk,uk,I)$ and $ok\leftarrow \mathrm{SecretVerify}(vk,uk,ck,U',V',I).$
By definition, the traceability experiment fails if either $ok=0$, or $ok=1$ but $U'\neq U$. Accordingly, define
\begin{equation}
W_1:=\{ok=0\},
\qquad
W_2:=\{ok=1 \wedge U'\neq U\}.
\end{equation}
Then
\begin{equation}
\Pr[\mathrm{ExpTrace}=1]=\Pr[W_1]+\Pr[W_2].
\end{equation}

For honestly generated auxiliary information, $\mathrm{SecretVerify}$ rejects the output of $\mathrm{SecretTrace}$ only with negligible probability. Hence
\begin{equation}
\Pr[W_1]\le \mathrm{negl}(\lambda).
\end{equation}

It remains to bound $\Pr[W_2]$. Write
\begin{equation}
W_2^{\mathrm{forge}}:=
\{\,ok=1 \wedge \exists\,tag^*\in U' \setminus Tags\,\},
\end{equation}
and
\begin{equation}
W_2^{\mathrm{wrong}}:=
\{\,ok=1 \wedge U'\subseteq Tags \wedge U'\neq U\,\}.
\end{equation}
Then
\begin{equation}
\Pr[W_2]=\Pr[W_2^{\mathrm{forge}}]+\Pr[W_2^{\mathrm{wrong}}].
\end{equation}

If $W_2^{\mathrm{forge}}$ occurs, then there exists $tag^*\in U'\setminus Tags$ that nevertheless passes verification. By the definition of $\mathrm{SecretVerify}$, this implies
\begin{equation}
H(tag^*)\in \{vk_i\}_{i=0}^{m-1},
\end{equation}
and hence $H(tag^*)=H(tag_j)$ for some $tag_j\in Tags$. Since $tag^*\neq tag_j$, this yields a collision in $H$. By collision resistance,
\begin{equation}
\Pr[W_2^{\mathrm{forge}}]\le \mathrm{negl}(\lambda).
\end{equation}

Next consider $W_2^{\mathrm{wrong}}$. In this case, $\mathrm{SecretVerify}$ accepts a subset of honestly generated tags that differs from the true honest tag set $U$. However, for honestly generated reconstruction data, each element of $I$ has the form $r_i\|tk_i$, and $\mathrm{SecretTrace}$ recovers the corresponding tag through the successive decryptions
\begin{equation}
r_i \xrightarrow{\mathrm{Dec}(\cdot,uk)} ct_i
\xrightarrow{\mathrm{Dec}(\cdot,ck)} tag_i.
\end{equation}
Thus, once $uk$ and $ck$ are fixed, the recovered tags are uniquely determined by $I$. Since $\mathrm{SecretVerify}$ accepts only if these recovered tags are consistent with the verification keys and the auxiliary data extracted from $I$, the accepted valid-tag set must coincide with $U$, except with negligible probability. Therefore,
\begin{equation}
\Pr[W_2^{\mathrm{wrong}}]\le \mathrm{negl}(\lambda).
\end{equation}

Combining the above bounds gives
\begin{equation}
\Pr[\mathrm{ExpTrace}=1]\le \mathrm{negl}(\lambda),
\end{equation}
and hence
\begin{equation}
\mathrm{Adv}^{\mathrm{trace}}_{\mathcal{A},\mathrm{MTA\text{-}QSS}}(\lambda)
\le \mathrm{negl}(\lambda).
\end{equation}
\end{proof}

\subsection{Proof of Theorem~\ref{th:ni2}}
\label{app:mta-ni2-proof}

\noindent\textbf{Theorem~\ref{th:ni2} }(Formal).
{\itshape
The MTA-QSS scheme is non-imputable against any QPT adversary. For every QPT adversary $\mathcal{A}$, there exists a negligible function $\mathrm{negl}(\cdot)$ such that, for security parameter $\lambda \in \mathbb{N}$,
\[
    \mathrm{Adv}^{\mathrm{ni}}_{\mathcal{A},\mathrm{MTA\text{-}QSS}}(\lambda)
    \leq \mathrm{negl}(\lambda).
\]
}

\begin{proof}
Consider the experiment $\mathrm{ExpNI}_{\mathcal{A},\mathrm{MTA\text{-}QSS}}(\lambda)$. Let $i^*\in [m]$ be the target honest participant selected by the adversary, and let $tag^*:=tag_{i^*}$ be the corresponding tag. Since participant $i^*$ remains honest throughout the experiment, $tag^*$ is an honestly generated tag and therefore satisfies
\begin{equation}
    H(tag^*)\in \{vk_i\}_{i=0}^{m-1}.
\end{equation}

Suppose the adversary outputs a pair $(U^*,V^*)$. The framing event is defined by
\begin{equation}
\begin{aligned}
    \mathrm{Frame}
    := \bigl\{\, &tag^* \in V^* \\
    &\wedge\ \mathrm{SecretVerify}(vk,uk,ck,U^*,V^*,I)=1 \,\bigr\}.
\end{aligned}
\end{equation}

However, by the definition of $\mathrm{SecretVerify}$, acceptance requires that
\begin{equation}
    \forall\, tag_z\in V^*,
    \qquad
    H(tag_z)\notin \{vk_i\}_{i=0}^{m-1}.
\end{equation}
This condition is incompatible with $tag^*\in V^*$, since $H(tag^*)\in \{vk_i\}_{i=0}^{m-1}$. Hence the event $\mathrm{Frame}$ cannot occur, and therefore
\begin{equation}
    \Pr[\mathrm{Frame}]=0.
\end{equation}

By the definition of non-imputability advantage,
\begin{equation}
    \mathrm{Adv}^{\mathrm{ni}}_{\mathcal{A},\mathrm{MTA\text{-}QSS}}(\lambda)
    =\Pr[\mathrm{Frame}]=0
    \le\mathrm{negl}(\lambda).
\end{equation}
Hence, the MTA-QSS scheme is non-imputable.
\end{proof}

\nocite{*}

\bibliography{apssamp}

\end{document}